\documentclass[10pt,journal]{IEEEtran}
\usepackage{amsfonts,amsmath,amssymb}
\usepackage{cite}
\usepackage{graphicx}
\usepackage{url}
\newtheorem{theorem}{Theorem}

  \newtheorem{remark}{Remark}
  
  \newtheorem{example}{Example}

\begin{document}

\title{On Achievable Degrees of Freedom for  MIMO X Channels}
\author{Lu~Yang,~\IEEEmembership{Student~Member,~IEEE,}~and~Wei~Zhang,~\IEEEmembership{Senior~Member,~IEEE}
\thanks{This paper was presented in part at International Conference on Wireless Communications and Signal Processing, Nanjing, China, Nov. 9-11, 2011.}
\thanks{L. Yang and W. Zhang are with School of Electrical Engineering \& Telecommunications, The University of New South Wales, Sydney, Australia
(e-mail: lu.yang@student.unsw.edu.au; wzhang@ee.unsw.edu.au). }
}

\maketitle

\begin{abstract}
 In this paper, the achievable DoF of  MIMO X channels for constant channel coefficients with $M_t$ antennas at transmitter $t$ and $N_r$ antennas at receiver $r$ ($t,r=1,2$) is studied.
A spatial interference alignment and cancelation scheme is proposed to achieve the maximum DoF of the MIMO X channels. The scenario of $M_1\geq M_2\geq N_1\geq N_2$ is first considered and divided into 3 cases, $3N_2<M_1+M_2<2N_1+N_2$ (Case $A$), $M_1+M_2\geq2N_1+N_2$ (Case $B$), and $M_1+M_2\leq3N_2$ (Case $C$).
  With the proposed scheme,  it is shown that in Case $A$, the outer-bound $\frac{M_1+M_2+N_2}{2}$ is achievable; in Case $B$, the achievable DoF equals the outer-bound $N_1+N_2$ if $M_2>N_1$, otherwise it is $\frac{1}{2}$ or $1$ less than the outer-bound; in Case $C$, the achievable DoF is equal to the outer-bound $\frac{2}{3}(M_1+M_2)$ if $(3N_2-M_1-M_2)\mod 3=0$, and it is $\frac{1}{3}$ or $\frac{1}{6}$ less than the outer-bound if $(3N_2-M_1-M_2)\mod 3=1~\mathrm{or}~2$.
In the scenario of $M_t\leq N_r$, the exact symmetrical results of DoF can be obtained.
\end{abstract}


\section{Introduction}

In recent years there is growing interest in capacity characterization of distributed wireless networks. In the high signal-to-noise ratio (SNR) regime, Degree of Freedom (DoF) provides accurate capacity approximation and offers fundamental insights into optimal interference management schemes\cite{1}.
The DoF benefits of overlapping interference space were first studied in \cite{MMK} for $2\times2$ X network, where an iterative algorithm was proposed for optimizing the transmitters and receivers in conjunction with dirty paper coding and successive decoding. It was shown in \cite{MMK} with $M$ antennas at each node totally $\lfloor\frac{4M}{3}\rfloor$ DoF was achieved. Afterward, the concept of \textit{interference alignment} was crystalized in \cite{MIMO} by Jafar and Shamai, where a closed-form solution for a beamforming scheme that achieves perfect interference alignment was provided.
The other setting of interference alignment is $K$-user interference channel \cite{K}, which further enhances the status of interference alignment as a general principle by establishing its applications in a variety of contexts, including propagation delay, phase alignment and beamforming.

The novel idea of interference alignment has challenged much of the conventional wisdom and has been then utilized in the DoF characterization of various system models, such as the $K$-user MIMO interference channel \cite{MIMOI,ICC2}, MIMO X channel\cite{ICC,X}, compound MISO BC channel \cite {BC}, down-link channel \cite{down,downlink}, etc. Although the benefits generated by interference alignment are remarkable, they have so far been shown mostly under idealized assumptions such as global channel knowledge and the need of channel variation. Some works have been done to deal with the former issue: \cite{approach,blind} try to implement interference alignment scheme with limited channel information at transmitter; \cite{delay,delayedi,delayedx} focus on the case in which the channel information is available at transmitters but has some delays, mostly due to the channel variations. It in fact leads us to the concerns of this paper -- the utilization of interference alignment schemes for constant or slow fading channels.

In this paper, we focus on the achievable DoF of  MIMO X channels with constant complex channel coefficients. Transmitter $t$ ($t=1,2$) is equipped with $M_t$ antennas and receiver $r$ ($r=1,2$) is equipped with $N_r$ antennas, denoted by ($M_1, M_2, N_1, N_2$).
We first review the related works that have been done in this area. The DoF of constant $2\times2$ MIMO X channels was first studied in \cite{MMK2}, in which some linear filters are employed at the transmitters and receivers to decompose the system into either two noninterfering multiple-antenna broadcast sub-channels or two noninterfering multiple-antenna multiple-access sub-channels. Then, with the use of spatial interference alignment, some surprisingly high DoF was obtained. In particular, it was shown in \cite{MMK2} that for systems of ($\lceil\frac{1}{2}\lfloor\frac{4N}{3}\rfloor\rceil$, $\lfloor\frac{1}{2}\lfloor\frac{4N}{3}\rfloor\rfloor$, $N$, $N$) and ($N$, $N$, $\lceil\frac{1}{2}\lfloor\frac{4N}{3}\rfloor\rceil$, $\lfloor\frac{1}{2}\lfloor\frac{4N}{3}\rfloor\rfloor$), the DoF of $\lfloor\frac{4N}{3}\rfloor$ can be achieved.
Afterward, \textit{signal level interference alignment } \cite{Tse,deter} was proposed, in which interference alignment is achieved in signal scale and through lattice codes. The idea was then further advanced and utilized in the DoF characterization of $K$-user interference channel \cite{real} and  MIMO X channels \cite{layer}.
In particular, a layered interference alignment scheme was proposed in \cite{layer} which utilized the concept of both vector alignment and signal alignment, combined with a number-theoretic joint processing technique at receivers. With the same number of antennas on each node, the outer-bound DoF can be achieved with real channel coefficients\cite{layer}. The process is backed up by a recent result in the field of Simultaneous Diophantine Approximation \cite{number}.
Recently, an effective technique called \textit{asymmetric signaling} was introduced in \cite{asy}, whose main idea is to explore the phase dimensions of communication system with asymmetric input. With the scheme proposed in \cite{asy}, optimal DoF can be achieved for a variety of single-antenna networks.

In this paper, we study the   MIMO X channels with constant complex channel coefficients, where each node is equipped with different number of antennas.
We propose an asymmetric interference alignment and cancelation scheme without symbol extension that achieves the outer-bound or near outer-bound DoF for both cases  $M_t\geq N_r$ and $M_t\leq N_r$ ($t,r=1,2$).
 In the scenario of $M_1\geq M_2\geq N_1\geq N_2$, it is divided into three cases, which are $3N_2<M_1+M_2<2N_1+N_2$ (Case $A$), $M_1+M_2\geq2N_1+N_2$ (Case $B$) and $M_1+M_2\leq3N_2$ (Case $C$). In each case, a linear optimization problem is formulated to maximize DoF. By solving the problem, the maximum achievable DoF can be determined. Specifically,
  in Case $A$, the outer-bound $\frac{M_1+M_2+N_2}{2}$ is achievable; in Case $B$, the achievable DoF equals the outer-bound $N_1+N_2$ if $M_2>N_1$, otherwise it is $\frac{1}{2}$ or $1$ less than the outer-bound; in Case $C$, the achievable DoF is equal to the outer-bound $\frac{2}{3}(M_1+M_2)$ if $(3N_2-M_1-M_2)\mod 3=0$, and it is $\frac{1}{3}$ or $\frac{1}{6}$ less than the outer-bound if $(3N_2-M_1-M_2)\mod 3=1~\mathrm{or}~2$.
  Moreover, an intuitive explanation is given for each case to validate the results.
 In the scenario of $M_t\leq N_r$, we show that  exact symmetrical results of DoF can be obtained.

The paper is organized as follows. In Section \ref{s2}, some main concepts incorporated in the scheme are presented. In Section \ref{result}, the system model and main results are introduced. In Section \ref{IAC}, an asymmetric interference alignment and cancelation scheme is described. In Section \ref{A},\ref{B}, and \ref{C}, the achievable DoF of the MIMO X channels for $M_t\geq N_r$ are investigated for Case $A$, $B$, and $C$, respectively. The DoF of $M_t\leq N_r$ is addressed in Section \ref{M<N}. Finally, Section \ref{con} concludes the paper.

\section{Main Concepts}\label{s2}

\subsection{Degrees of Freedom}
The DoF of message $m$ transmitted in the system is defined as \cite{X}
\begin{eqnarray}
  d_m=\lim_{\rho \to \infty}\frac{R_m(\rho)}{\log_2{\rho}}
\end{eqnarray}
where $\rho$ denotes the power constraint of the message and $R_m(\rho)$ represents the rate of the codeword encoding the message $m$. Consider a single user point-to-point channel where the transmitted constellation $\mathcal{U}(-Q,Q)_\mathbb{Z}=\{-Q,-Q+1,\ldots,-1,1,\dots,Q-1,Q\}$ ($Q$ is an integer) is used for a single message. Since it is assumed that the additive noise has unit variance and the minimum distance in the received constellation is, the same as transmitted constellation, also one, the noise can be treated as removable \cite{real}. Therefore $R_m\approx{\log{2Q}}$ is achievable for the channel. In addition, the power constraint should be no less than $Q^2$. Hence, $\rho=Q^2$, and the DoF associated with the message can be calculated as
\begin{eqnarray}
  d_m=\lim_{Q \to \infty}\frac{R_m=\log{(2Q)}}{\log_2{Q^2}}=\frac{1}{2}\label{1/2}
\end{eqnarray}

If the message ($m=u+\mathbf{j}v$) is modulated with a two-dimensional constellation $\mathcal{U}=\mathcal{V}=(-Q,Q)_\mathbb{Z}=\{-Q,-Q+1,\ldots,-1,1,\dots,Q-1,Q\}$ , the rate will become $R_m=2\log{(2Q)}$. Since the power constraint is $2Q^2$, each message will carry $1$ DoF, i.e.,
\begin{eqnarray}
  d_m=\lim_{Q \to \infty}\frac{R_m=2\log{(2Q)}}{\log_2{2Q^2}}=1
\end{eqnarray}

As we can see, if the message is a complex number and has both real and imaginary parts, the total DoF is the sum DoF of each part.

\subsection{Structured Coding}

In this paper, it is assumed that each message has only one dimension (real). Given that two-dimensional constellation is much more common in practical modulation schemes (such as QAM), we propose a coding scheme such that the complex message ($m=u+\mathbf{j}v$) can be transformed into a real number $s$. We let
\begin{eqnarray}
  s=u+c\cdot v
\end{eqnarray}
where $c$ is an integer. Since the sum of two structured codes is still a structured code, $s$ will have the constellation of $\mathcal{U}'$.\footnote[1]{$\mathcal{U}'=(-cQ-Q,cQ+Q)_\mathbb{Z}=\{-cQ-Q,-cQ-Q+1,\dots,-cQ+Q,-c(Q-1)-Q,\dots,-c+Q,c-Q,c-Q+1,\dots,c(Q-1)+Q,cQ-Q,\dots,cQ+Q-1,cQ+Q\}$}
To guarantee each point in this constellation does not overlap with others and keep the minimum distance equal to or larger than one, $c$ must satisfy $c\geq2Q+1$. By doing this, there would be a one-to-one mapping from the real number $s$ to the original message $m$. For example, if the message $m$ is modulated with QPSK, then $Q=1$, and $m$ must be one of the following four points $\{-1-\mathbf{j}, 1-\mathbf{j}, -1+\mathbf{j}, 1+\mathbf{j}\}$. If we let $c=2Q+1=3$, the constellation of $s$ would be $\{-4,-2,2,4\}$.

Therefore, the assumption of messages being real does not lose its generality. The price we pay here is that the power constraint is no longer $Q^2$, but $(cQ)^2+Q^2$. Since $c=2Q+1$, the DoF of $s$ is calculated as
\begin{eqnarray}
  d_s=\lim_{Q \to \infty}\frac{2\log{(2Q)}}{\log_2{((cQ)^2+Q^2)}}=\frac{1}{2}
\end{eqnarray}

\subsection{Asymmetric Signaling}
\setcounter{equation}{8}
\begin{figure*}
\begin{eqnarray}
            \bar{H}=\left[\begin{array}{cccc}
                               |h^{11}|\cos{\varphi^{11}} &-|h^{11}|\sin{\varphi^{11}} & |h^{12}|\cos{\varphi^{12}} &-|h^{12}|\sin{\varphi^{12}} \\
                               |h^{11}|\sin{\varphi^{11}} & |h^{11}|\cos{\varphi^{11}} & |h^{12}|\sin{\varphi^{12}} & ~|h^{12}|\cos{\varphi^{12}} \\
                               |h^{21}|\cos{\varphi^{21}} &-|h^{21}|\sin{\varphi^{21}} & |h^{22}|\cos{\varphi^{22}} &-|h^{22}|\sin{\varphi^{22}} \\
                               |h^{21}|\sin{\varphi^{21}} & |h^{21}|\cos{\varphi^{21}} & |h^{22}|\sin{\varphi^{22}} & |h^{22}|\cos{\varphi^{22}}
                             \end{array}\right]\label{H}
          \end{eqnarray}
          \hrulefill
\end{figure*}
\setcounter{equation}{5}
In wireless communication, we normally come across symmetric complex Gaussian variables such as additive noise, fading channels, and so are the input signals, whose real and imaginary parts are independent of each other. Inspired by\cite{asy}, we use asymmetric input in our scheme, in which the input signals are chosen to be complex but not symmetric. By doing so,  an $M$-dimensional complex system can be transformed into a $2M$-dimensional real system.

For instance, we consider a MIMO point-to-point channel with two antennas at each side. Let $\mathbf{x}\in \mathbb{C}^{2\times1}$ denote the transmitted signal and $\mathbf{y}\in \mathbb{C}^{2\times1}$ denote the received signal. We have
\begin{eqnarray}
  \mathbf{y}&=&\left[\begin{array}{c}
               y_1 \\
               y_2
             \end{array}\right]
   = \left[\begin{array}{cc}
                    h^{11} & h^{12} \\
                    h^{21} & h^{22}
                  \end{array}\right] \underbrace{\mathbf{v}\cdot m}_{\mathbf{x}} \label{complex}
\end{eqnarray}
where $\mathbf{v}=\left[\begin{array}{cc}
                    v_1 & v_2
                  \end{array}
\right]^\mathbf{T}$ denotes the precoding vector, $m$ is the original real message, and $h^{ij}$ denotes the channel gain from the $j$th transmit antenna to the $i$th receive antenna with phase $\varphi^{ij}$, which can be written as
\begin{eqnarray}
  h^{ij} &=& |h^{ij}|(\cos{\varphi^{ij}}+\mathbf{j}\sin{\varphi^{ij}}).
\end{eqnarray}

Therefore, (\ref{complex}) can be expressed alternatively as a real system, i.e.,
\begin{eqnarray}
  \bar{Y} =\left[\begin{array}{c}
             \mathrm{Re}(y_1) \\
             \mathrm{Im}(y_1) \\
             \mathrm{Re}(y_2) \\
             \mathrm{Im}(y_2)
           \end{array}\right]=\bar{H}\underbrace{\left[\begin{array}{c}
             \mathrm{Re}(v_1) \\
             \mathrm{Im}(v_1) \\
             \mathrm{Re}(v_2) \\
             \mathrm{Im}(v_2)
           \end{array}\right]}_{\bar{V}}m \label{real}
\end{eqnarray}
where $\mathrm{Re}(v)$ and $\mathrm{Im}(v)$ denote real and imaginary parts of $v$, respectively, and the equivalent channel matrix $\bar{H}$ is expressed as (\ref{H}).

It can be seen that the $2\times2$ complex system is turned into a $4\times4$ real system.
\setcounter{equation}{9}

\section{System Model and Main Result}\label{result}

\subsection {System Model}

We consider a $2\times2$ MIMO X network as depicted in Fig. \ref{fig1}. Transmitter $T_t$ ($t=1,2$) is equipped with $M_t$ antennas and receiver $R_r$ ($r=1,2$) is equipped with $N_r$ antennas. This configuration of antennas is denoted by ($M_1, M_2, N_1, N_2$). Without loss of generality, we assume that $M_1\geq M_2$ and $N_1\geq N_2$.

\begin{figure}[t!]
  \begin{center}
        \includegraphics[width=0.4\columnwidth]{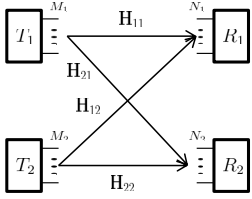}
        \caption{$2\times2$ MIMO $X$ channel $(M_1,~M_2,~N_1,~N_2)$}
        \label{fig1}
    \end{center}
\end{figure}

Let $h_{rt}^{ij}$ denote the channel gain from the $j$th antenna of transmitter $t$ to the $i$th antenna of receiver $r$. It can be expressed as
\begin{eqnarray}
  h_{rt}^{ij} &=& |h_{rt}^{ij}|(\cos{\varphi_{rt}^{ij}}+\mathbf{j}\sin{\varphi_{rt}^{ij}})
\end{eqnarray}
where $\varphi_{rt}^{ij}$ denotes the phase of $h_{rt}^{ij}$.

With asymmetric signaling, we can let $\mathbf{H}_{rt}$ denote the channel matrix between transmitter $t$ and receiver $r$ and let $\bar{H}_{rt}$ denote its alternative form with real quantities. All the channel matrices are sampled from continuous complex Gaussian distributions and each entry of $\mathbf{H}_{rt}$ is independent and identically distributed (i.i.d.).  The global channel information is assumed to be available at all nodes.

Let $\mathbf{m}_{rt}$ denote the message vector intended for receiver $r$ from transmitter $t$. With the proposed structured coding method, all elements of $\mathbf{m}_{rt}$ (the original messages $m_{rt}$) are set to be real, and each carries $\frac{1}{2}$ DoF according to (\ref{1/2}).

\subsection {Main Results}

The outer-bound DoF of the MIMO X channels was derived in \cite[Eq. (26)]{MIMO}, whose forms are different according to various settings of antenna number on each node. In this paper, we propose a signal transmission scheme that approaches the outer-bound or near outer-bound for both cases $M_t\geq N_r$ and $M_t\leq N_r$ ($t,r=1,2$).

\begin{table*}[t!]
\begin{center}\caption{Achievable DoF of MIMO X Channels ($M_1\geq M_2\geq N_1\geq N_2$)}
\renewcommand{\arraystretch}{2 }
\centering\label{table1}
\begin{tabular} {c |c |c |c }
\hline \hline
 Case  & Antenna configuration & Outer-bound DoF & Achievable DoF \\
\hline
$A $ & $3N_2<M_1+M_2<2N_1+N_2$ & $\frac{M_1+M_2+N_2}{2}$ & $\frac{M_1+M_2+N_2}{2}$ \\
\hline
$B$ & $M_1+M_2\geq2N_1+N_2$ & $N_1+N_2$ & $[N_1+N_2-1,N_1+N_2]$\\
\hline
$C$ &$ M_1+M_2\leq3N_2 $ & $ \frac{2}{3}(M_1+M_2) $ & $[\frac{2}{3}(M_1+M_2)-\frac{1}{3},\frac{2}{3}(M_1+M_2)]$ \\
\hline\hline
\end{tabular}
\end{center}
\end{table*}

\emph{Result 1}: When the number of transmitter antennas is larger than or equal to that of receiver antennas ($M_t\geq N_r$), it can be divided into three cases (as shown in Table \ref{table1}). An asymmetric interference alignment and cancelation scheme is proposed in Section \ref{IAC} that achieves the outer-bound or near outer-bound of MIMO X channels. Specifically, for $3N_2<M_1+M_2<2N_1+N_2$ (Case $A$), the exact outer-bound can be achieved.
For $M_1+M_2\geq2N_1+N_2$ (Case $B$), the outer-bound can be achieved for $M_2>N_1$. If $M_2=N_1$, to maintain the structure of the network as an X channel (not a broadcast or Z channel), the achievable DoF is $\frac{1}{2}$ or $1$ less than the outer-bound.
For $ M_1+M_2\leq3N_2$ (Case $C$), the achievable DoF is equal to the outer-bound if $(3N_2-M_1-M_2)\mod 3=0$, and it is $\frac{1}{3}$ or $\frac{1}{6}$ less than the outer-bound if $(3N_2-M_1-M_2)\mod 3=1~\mathrm{or}~2$. The achievable DoF of Cases $A$, $B$, and $C$ are proved in Section \ref{A},\ref{B}, and \ref{C}, respectively.

\emph{Result 2}:  When the number of transmitter antennas is smaller than or equal to that of receiver antennas ($M_t\leq N_r$), it can also be divided into three cases (as shown in Table \ref{table5}).  We propose an interference alignment-based precoding scheme to achieve the outer-bound or near outer-bound of MIMO X channels. It can be seen, the achievable DoF in this scenario is exactly symmetrical to \emph{Result 1}. The scheme and the proof of the results are given in Section \ref{M<N}.

\begin{table*}[t!]
\begin{center}\caption{Achievable DoF of MIMO X Channels ($N_1\geq N_2\geq M_1\geq M_2$)}
\renewcommand{\arraystretch}{2 }
\centering\label{table5}
\begin{tabular} {c |c |c |c}
\hline \hline
 Case  & Antenna configuration & Outer-bound DoF & Achievable DoF\\
\hline
$A' $ & $3M_2<N_1+N_2<2M_1+M_2$ & $\frac{N_1+N_2+M_2}{2}$ & $\frac{N_1+N_2+M_2}{2}$\\
\hline
$B'$ & $N_1+N_2\geq2M_1+M_2$ & $M_1+M_2$ & $[M_1+M_2-1,~M_1+M_2]$\\
\hline
$C'$ &$ N_1+N_2\leq3M_2 $ & $\frac{2}{3}(N_1+N_2)$ & $[\frac{2}{3}(N_1+N_2)-\frac{1}{3},~\frac{2}{3}(N_1+N_2)]$\\
\hline\hline
\end{tabular}
\end{center}
\hrulefill
\end{table*}

\section{Asymmetric Interference Alignment and Cancelation Scheme}\label{IAC}

In this section, we first elaborate the designs of transmitted signals and their precoding vectors in the scenario of $M_t\geq N_r$. Then, we show that the signals at each receiver are independent of each other.

\subsection{Design of Transmitted Signals}


$\bullet$ \textbf{Transmitted signal at $T_1$}

There are two message vectors $\mathbf{m}_{11}$ and $\mathbf{m}_{21}$ at $T_1$, which are desired signals of $R_1$ and $R_2$, respectively.

For $\mathbf{m}_{11}$, it has three blocks  $\mathbf{m}_{11}^1$, $\mathbf{m}_{11}^2$ and $\mathbf{m}_{11}^3$, each having length $L_1$, $L_2$ and $L_3$, respectively, i.e.,
\begin{eqnarray}
  \mathbf{m}_{11}=\left[\begin{array}{ccc}
                    \underbrace{(\mathbf{m}_{11}^1)^\mathbf{T}}_{L_1} & \underbrace{(\mathbf{m}_{11}^2)^\mathbf{T}}_{L_2} & \underbrace{(\mathbf{m}_{11}^3)^\mathbf{T}}_{L_3}
                  \end{array}\right]^\mathbf{T} \label{m11}
\end{eqnarray}

Let $Q_{rt}$ denote the length of $\mathbf{m}_{rt}$ ($\mathbf{m}_{rt} \in \mathbb{R}^{Q_{rt}\times1}$), we have
\begin{equation}\label{q11}
    L_1+L_2+L_3=Q_{11}.
\end{equation}

Further,  $\mathbf{m}_{11}^1$ is precoded with $[\mathbf{v}_{11}^1\cdots \mathbf{v}_{11}^{L_1}]\in\mathbb{C}^{M_1\times L_1}$, $\mathbf{m}_{11}^2$ is precoded with $[\mathbf{w}_{11}^1\cdots \mathbf{w}_{11}^{L_2}]\in\mathbb{C}^{M_1\times L_2}$, and $\mathbf{m}_{11}^3$ is precoded with $[\mathbf{u}_{11}^1\cdots \mathbf{u}_{11}^{L_3}]\in\mathbb{C}^{M_1\times L_3}$. Then, the transmitted signal intended for $R_1$ from $T_1$ can be expressed as
\begin{eqnarray}
  \mathbf{x}_{11} =\underbrace{\left[\begin{array}{ccc}
                     \mathbf{v}_{11}^1& \cdots & \mathbf{v}_{11}^{L_1}
                  \end{array}\right]\mathbf{m}_{11}^1}_{\mathbf{x}_{11}^1}
                  +\underbrace{\left[\begin{array}{ccc}
                     \mathbf{w}_{11}^1& \cdots & \mathbf{w}_{11}^{L_2}
                  \end{array}\right]\mathbf{m}_{11}^2}_{\mathbf{x}_{11}^2}\nonumber\\
                  +\underbrace{\left[\begin{array}{ccc}
                     \mathbf{u}_{11}^1& \cdots & \mathbf{u}_{11}^{L_3}
                  \end{array}\right]\mathbf{m}_{11}^3}_{\mathbf{x}_{11}^3}~~~~~~~~~~~~~~~~~~~~~~~~~~~~~~~~~~\nonumber
\end{eqnarray}

Similarly, we divide $\mathbf{m}_{21}$ into three blocks $\mathbf{m}_{21}^1$, $\mathbf{m}_{21}^2$, and
$\mathbf{m}_{21}^3$, each having length $K_1$, $K_2$ and $K_3$, respectively, i.e.,
\begin{eqnarray}
  \mathbf{m}_{21}=\left[\begin{array}{ccc}
                    \underbrace{(\mathbf{m}_{21}^1)^\mathbf{T}}_{K_1} & \underbrace{(\mathbf{m}_{21}^2)^\mathbf{T}}_{K_2} & \underbrace{(\mathbf{m}_{21}^3)^\mathbf{T}}_{K_3}
                  \end{array}\right]^\mathbf{T} \label{m21}
\end{eqnarray}
and
\begin{equation}\label{q21}
    K_1+K_2+K_3=Q_{21}.
\end{equation}

Furthermore,  $\mathbf{m}_{21}^1$ is precoded with $[\mathbf{v}_{21}^1\cdots \mathbf{v}_{21}^{K_1}]$, $\mathbf{m}_{21}^2$ is precoded with $[\mathbf{w}_{21}^1\cdots \mathbf{w}_{21}^{K_2}]$, and $\mathbf{m}_{21}^3$ is precoded with $[\mathbf{u}_{21}^1\cdots \mathbf{u}_{21}^{K_3}]$, respectively. Then, the transmitted signal intended to $R_2$ from $T_1$ can be written as
\begin{eqnarray}
  \mathbf{x}_{21}=\underbrace{\left[\begin{array}{ccc}
                     \mathbf{v}_{21}^1& \cdots & \mathbf{v}_{21}^{K_1}
                  \end{array}\right] \mathbf{m}_{21}^1}_{\mathbf{x}_{21}^1}
                  +\underbrace{\left[\begin{array}{ccc}
                     \mathbf{w}_{21}^1& \cdots & \mathbf{w}_{21}^{K_2}
                  \end{array}\right] \mathbf{m}_{21}^2}_{\mathbf{x}_{21}^2}\nonumber\\
                  +\underbrace{\left[\begin{array}{ccc}
                     \mathbf{u}_{21}^1& \cdots & \mathbf{u}_{21}^{K_3}
                  \end{array}\right] \mathbf{m}_{21}^3}_{\mathbf{x}_{21}^3}~~~~~~~~~~~~~~~~~~~~~~~~~~~~~~~~~~\nonumber
\end{eqnarray}

$\bullet$ \textbf{Transmitted signal at $T_2$}

At $T_2$, two message vectors $\mathbf{m}_{12}$ and $\mathbf{m}_{22}$ will be sent, which are the desired signals of $R_1$ and $R_2$, respectively.

For $\mathbf{m}_{12}$, it is also divided into three blocks $\mathbf{m}_{12}^1$, $\mathbf{m}_{12}^2$ and $\mathbf{m}_{12}^3$, each having length $J_1$, $J_2$ and $J_3$ respectively, i.e.,
\begin{eqnarray}
  \mathbf{m}_{12}=\left[\begin{array}{ccc}
                    \underbrace{(\mathbf{m}_{12}^1)^\mathbf{T}}_{J_1} & \underbrace{(\mathbf{m}_{12}^2)^\mathbf{T}}_{J_2} & \underbrace{(\mathbf{m}_{12}^3)^\mathbf{T}}_{J_3}
                  \end{array}\right]^\mathbf{T} \label{m12}
\end{eqnarray}
and
\begin{equation}\label{q12}
    J_1+J_2+J_3=Q_{12}
\end{equation}

We let $\mathbf{m}_{12}^1$ be precoded with $[\mathbf{v}_{12}^1\cdots \mathbf{v}_{12}^{J_1}]$; $\mathbf{m}_{12}^2$ is precoded with $[\mathbf{w}_{12}^1\cdots \mathbf{w}_{12}^{J_2}]$; and $\mathbf{m}_{12}^3$ is precoded with $[\mathbf{u}_{12}^1\cdots \mathbf{u}_{12}^{J_3}]$. Then, the transmitted signal from $T_2$ intended to $R_1$ can be expressed as
\begin{eqnarray}\label{x12}
  \mathbf{x}_{12} =\underbrace{\left[\begin{array}{ccc}
                     \mathbf{v}_{12}^1& \cdots & \mathbf{v}_{12}^{J_1}
                  \end{array}\right]\mathbf{m}_{12}^1}_{\mathbf{x}_{12}^1}
                  +\underbrace{\left[\begin{array}{ccc}
                     \mathbf{w}_{12}^1& \cdots & \mathbf{w}_{12}^{J_2}
                  \end{array}\right]\mathbf{m}_{12}^2}_{\mathbf{x}_{12}^2}\nonumber\\
                  +\underbrace{\left[\begin{array}{ccc}
                     \mathbf{u}_{12}^1& \cdots & \mathbf{u}_{12}^{J_3}
                  \end{array}\right]\mathbf{m}_{12}^3}_{\mathbf{x}_{12}^3}~~~~~~~~~~~~~~~~~~~~~~~~~~~~~~~~~\nonumber
\end{eqnarray}

For $\mathbf{m}_{22}$, we divide the message vector into three blocks $\mathbf{m}_{22}^1$, $\mathbf{m}_{22}^2$ and $\mathbf{m}_{22}^3$, each having length $G_1$, $G_2$ and $G_3$, respectively, i.e.,
\begin{eqnarray}
  \mathbf{m}_{22}=\left[\begin{array}{ccc}
                    \underbrace{(\mathbf{m}_{22}^1)^\mathbf{T}}_{G_1} & \underbrace{(\mathbf{m}_{22}^2)^\mathbf{T}}_{G_2} & \underbrace{(\mathbf{m}_{22}^3)^\mathbf{T}}_{G_3}
                  \end{array}\right]^\mathbf{T} \label{m22}
\end{eqnarray}
and
\begin{equation}\label{q22}
    G_1+G_2+G_3=Q_{22}
\end{equation}

We let $\mathbf{m}_{22}^1$ be precoded with $[\mathbf{v}_{22}^1\cdots \mathbf{v}_{22}^{G_1}]$; $\mathbf{m}_{22}^2$ is precoded with $[\mathbf{w}_{22}^1\cdots \mathbf{w}_{22}^{G_2}]$, and $\mathbf{m}_{22}^3$ is precoded with $[\mathbf{u}_{22}^1\cdots \mathbf{u}_{22}^{G_3}]$, respectively.
Then, the transmitted signal intended to $R_2$ from $T_2$ can be expressed as
\begin{eqnarray}
  \mathbf{x}_{22} =\underbrace{\left[\begin{array}{ccc}
                     \mathbf{v}_{22}^1& \cdots & \mathbf{v}_{22}^{G_1}
                  \end{array}\right] \mathbf{m}_{22}^1}_{\mathbf{x}_{22}^1}
                  +\underbrace{\left[\begin{array}{ccc}
                     \mathbf{w}_{22}^1& \cdots & \mathbf{w}_{22}^{G_2}
                  \end{array}\right] \mathbf{m}_{22}^2}_{\mathbf{x}_{22}^2}\nonumber\\
                  +\underbrace{\left[\begin{array}{ccc}
                     \mathbf{u}_{22}^1& \cdots & \mathbf{u}_{22}^{G_3}
                  \end{array}\right] \mathbf{m}_{22}^3}_{\mathbf{x}_{22}^3}~~~~~~~~~~~~~~~~~~~~~~~~~~~~~~~~~~\nonumber
\end{eqnarray}

If all desired signals are independent of each other at each receiver, the total DoF of the system can be calculated as
\begin{eqnarray}
  D_{sum}=\frac{Q_{11}+Q_{21}+Q_{12}+Q_{22}}{2} \label{Dsum}
\end{eqnarray}

\subsection{Design of Precoding Vectors}

We first examine the received signals at $R_1$. It can be expressed as
\begin{eqnarray}
   \mathbf{Y}_1&=& \mathbf{H}_{11}(\mathbf{x}_{11}^1+\mathbf{x}_{11}^2+\mathbf{x}_{11}^3)+\mathbf{H}_{12}(\mathbf{x}_{12}^1+\mathbf{x}_{12}^2+\mathbf{x}_{12}^3)\nonumber\\
               &+&\underbrace{\mathbf{H}_{11}(\mathbf{x}_{21}^1+\mathbf{x}_{21}^2+\mathbf{x}_{21}^3)+\mathbf{H}_{12}(\mathbf{x}_{22}^1+\mathbf{x}_{22}^2+\mathbf{x}_{22}^3)}_{\rm{interferece}}+\mathbf{z}_1\nonumber\\\label{eqn:R1}
 \end{eqnarray}
where $\mathbf{z}_r$ denotes the white noise vector at receiver $r$. Each entry of $\mathbf{z}_r$ is i.i.d. with $\mathcal{CN}(0.~1)$.

 It can be seen that $\mathbf{x}_{21}^1$, $\mathbf{x}_{21}^2$, $\mathbf{x}_{21}^3$, and $\mathbf{x}_{22}^1$, $\mathbf{x}_{22}^2$, $\mathbf{x}_{22}^3$ are the desired signals for $R_2$, but also the interference for $R_1$.

To null out the interference $\mathbf{x}_{21}^1$, $\mathbf{x}_{21}^2$ and $\mathbf{x}_{22}^1$, $\mathbf{x}_{22}^2$ at $R_1$,  we can let
\begin{eqnarray}
  \mathbf{H}_{11}\left[\begin{array}{ccc}
                     \mathbf{v}_{21}^1& \cdots & \mathbf{v}_{21}^{K_1}
                  \end{array}\right]=\mathbf{0}\nonumber \\
                  \mathbf{H}_{11}\left[\begin{array}{ccc}
                     \mathbf{w}_{21}^1& \cdots & \mathbf{w}_{21}^{K_2}
                  \end{array}\right]=\mathbf{0}
\end{eqnarray}
and
\begin{eqnarray}
  \mathbf{H}_{12}\left[\begin{array}{ccc}
                     \mathbf{v}_{22}^1& \cdots & \mathbf{v}_{22}^{G_1}
                  \end{array}\right]=\mathbf{0}\nonumber \\
                  \mathbf{H}_{12}\left[\begin{array}{ccc}
                     \mathbf{w}_{22}^1& \cdots & \mathbf{w}_{22}^{G_2}
                  \end{array}\right]=\mathbf{0}
\end{eqnarray}

These can be achieved by letting
\begin{eqnarray}
  \mathbf{v}_{21}^1, \cdots, \mathbf{v}_{21}^{K_1}&\subset& \mathrm{span}\{\mathbf{P}_{11}\}\label{v21}\\
  \mathbf{w}_{21}^k&=&\mathbf{j}\cdot\mathbf{v}_{21}^k, ~~~~k=1, 2, \cdots, K_2.\label{w21}\\
  \mathbf{v}_{22}^1, \cdots, \mathbf{v}_{22}^{G_1}&\subset& \mathrm{span}\{\mathbf{P}_{12}\}\label{v22}\\
  \mathbf{w}_{22}^g&=&\mathbf{j}\cdot\mathbf{v}_{22}^g, ~~~~g=1, 2, \cdots, G_2\label{w22}
\end{eqnarray}
where $\mathbf{P}_{rt}$ denotes the null space of $\mathbf{H}_{rt}$.

For each channel matrix $\mathbf{H}_{rt}$, there are $M_t-N_r$ independent column vectors in its null space $\mathbf{P}_{rt}$. In order to satisfy (\ref{v21}) to (\ref{w22}), we can set
\begin{eqnarray}
 K_2&\leq& K_1 \leq M_1-N_1\label{K1}\\
 G_2&\leq& G_1 \leq M_2-N_1\label{G1}
\end{eqnarray}

In addition, we want each signal of $\mathbf{x}_{22}^3$ to be aligned with one signal of $\mathbf{x}_{21}^3$ at $R_1$ in real space.  This can be done by letting
\begin{eqnarray}
  \bar{H}_{12}\bar{U}_{22}^{k_3}&\subseteq& \mathrm{span}\{\bar{H}_{11}\bar{U}_{21}^{k_3}\} ~~k_3=1, 2 \cdots K_3\label{G3} \\
  K_3&=&G_3\label{K3}\\
  \mathbf{u}_{21}^1, \cdots, \mathbf{u}_{21}^{K_3}&\nsubseteq& \mathrm{span}\{\mathbf{P}_{11}\}\label{u21}
 \end{eqnarray}

Since $2N_1\leq 2M_2$, $\bar{U}_{22}^{k_3}$ can always be found to achieve (\ref{G3}). Note that if two signals are aligned in real space, they are also aligned in complex space (it does not hold otherwise). Then, (\ref{u21}) is to guarantee that $\mathbf{x}_{21}^3$ is independent of $\mathbf{x}_{21}^1$ and $\mathbf{x}_{21}^2$.

Now, the precoding vectors of the signals intended to $R_2$ can be determined accordingly. Specifically, we pick $K_1$ independent vectors from the null space of $\mathbf{H}_{11}$ as precoders $\left[
  \begin{array}{ccc}
    \mathbf{v}_{21}^1 & \cdots & \mathbf{v}_{21}^{K_1} \\
  \end{array}
\right]$. Then, the precoders  $\left[
  \begin{array}{ccc}
    \mathbf{w}_{21}^1 & \cdots & \mathbf{w}_{21}^{K_2} \\
  \end{array}
\right]$ can be determined according to (\ref{w21}). The precoders $\left[
  \begin{array}{ccc}
    \mathbf{v}_{22}^1 & \cdots & \mathbf{v}_{22}^{G_1} \\
  \end{array}
\right]$ and $\left[
  \begin{array}{ccc}
    \mathbf{w}_{22}^1 & \cdots & \mathbf{w}_{22}^{G_2} \\
  \end{array}
\right]$ can be chosen based on (\ref{v22}) and (\ref{w22}), respectively.
Further, we choose $K_3$ independent vectors that satisfy (\ref{u21}) as the precoders $\left[
  \begin{array}{ccc}
    \mathbf{u}_{21}^1 & \cdots & \mathbf{u}_{21}^{K_3} \\
  \end{array}
\right]$, which means $K_3$ must be no larger than the rank of $\mathbf{H}_{11}$, i.e.,
\begin{eqnarray}
  K_3 \leq N_1 \label{K3N1}
\end{eqnarray}
Finally, the precoders $\left[
  \begin{array}{ccc}
    \mathbf{u}_{22}^1 & \cdots & \mathbf{u}_{22}^{G_3} \\
  \end{array}
\right]$ can be determined based on (\ref{G3}).

The received signal at $R_2$ can be expressed as
\begin{eqnarray}
   \mathbf{Y}_2&=& \nonumber\mathbf{H}_{21}(\mathbf{x}_{21}^1+\mathbf{x}_{21}^2+\mathbf{x}_{21}^3)+\mathbf{H}_{22}(\mathbf{x}_{22}^1+\mathbf{x}_{22}^2+\mathbf{x}_{22}^3)\\
               &+&\underbrace{\mathbf{H}_{21}(\mathbf{x}_{11}^1+\mathbf{x}_{11}^2+\mathbf{x}_{11}^3)+\mathbf{H}_{22}(\mathbf{x}_{12}^1+\mathbf{x}_{12}^2+\mathbf{x}_{12}^3)}_{\rm{interferece}}+\mathbf{z}_2\nonumber\\\label{eqn:R2}
\end{eqnarray}
In order to  null out $\mathbf{x}_{11}^1$, $\mathbf{x}_{11}^2$ and $\mathbf{x}_{12}^1$, $\mathbf{x}_{12}^2$ at $R_2$, we let
\begin{eqnarray}
  \mathbf{H}_{21}\left[\begin{array}{ccc}\label{v11}
                     \mathbf{v}_{11}^1& \cdots & \mathbf{v}_{11}^{L_1}
                  \end{array}\right]=\mathbf{0}\\
                  \mathbf{H}_{21}\left[\begin{array}{ccc}
                     \mathbf{w}_{11}^1& \cdots & \mathbf{w}_{11}^{L_2}
                  \end{array}\right]=\mathbf{0}
\end{eqnarray}
and
\begin{eqnarray}
  \mathbf{H}_{22}\left[\begin{array}{ccc}
                     \mathbf{v}_{12}^1& \cdots & \mathbf{v}_{12}^{J_1}
                  \end{array}\right]=\mathbf{0}\label{v12}\\
                  \mathbf{H}_{22}\left[\begin{array}{ccc}
                     \mathbf{w}_{12}^1& \cdots & \mathbf{w}_{12}^{J_2}
                  \end{array}\right]=\mathbf{0}\label{v12b}
\end{eqnarray}


These can be achieved by letting
\begin{eqnarray}
  \mathbf{v}_{11}^1, \cdots, \mathbf{v}_{11}^{L_1}&\subset& \rm{span}\{\mathbf{P}_{21}\}\label{L11}\\
  \mathbf{w}_{11}^l&=&\mathbf{j}\cdot\mathbf{v}_{11}^l, ~~l=1, 2,\cdots, L_2 \label{eqn:w11align}\\
   \mathbf{v}_{12}^1, \cdots, \mathbf{v}_{12}^{J_1} &\subset& \rm{span}\{\mathbf{P}_{22}\}\label{v12null}\\
  \mathbf{w}_{12}^{j_2}&=&\mathbf{j}\cdot\mathbf{v}_{12}^{j_2}, ~~j_2=1, 2, \cdots, J_2 \label{w12align}
\end{eqnarray}
which lead to
\begin{eqnarray}
  L_2&\leq& L_1\leq M_1-N_2 \label{L12}\\
   J_2 &\leq& J_1 \leq M_2-N_2 \label{J2}
\end{eqnarray}

Further, we want each signal of $\mathbf{x}_{11}^3$ to be aligned with one signal of $\mathbf{x}_{12}^3$ at $R_2$ in real space, i.e.,
\begin{eqnarray}
  \bar{H}_{22}\bar{U}_{12}^{l_3} &=& \mathrm{span}\{\bar{H}_{21}\bar{U}_{11}^{l_3}\}, ~~l_3=1, 2 \cdots L_3 \label{L3c} \\
   \mathbf{u}_{11}^1, \cdots, \mathbf{u}_{11}^{L_3}&\nsubseteq& \mathrm{span}\{\mathbf{P}_{21}\}\label{p21}\\
    J_3&=&L_3\leq N_2 \label{J3c}
\end{eqnarray}

Therefore, the precoding vectors of the signals intended to $R_1$ can be determined in the same way as those of the signals intended to $R_2$.

\subsection{Proof of Signal Independence}


We first examine the received signals on $R_2$, which can be expressed as (\ref{Y2}),
\begin{figure*}
\begin{eqnarray}
  \mathbf{Y}_2&=&\mathbf{H}_{21}(\mathbf{x}_{21}^1+\mathbf{x}_{21}^2+\mathbf{x}_{21}^3)+\mathbf{H}_{22}(\mathbf{x}_{22}^1+\mathbf{x}_{22}^2+\mathbf{x}_{22}^3)
  +\underbrace{\mathbf{H}_{21}\mathbf{x}_{11}^3+\mathbf{H}_{22}\mathbf{x}_{12}^3}_{\mathrm{interference}}+\mathbf{z}_2\nonumber\\
  &=&\underbrace{ \mathbf{H}_{21}[\mathbf{v}_{21}^1, \cdots,  \mathbf{v}_{21}^{K_2}]\mathbf{m}_{21}^1(1:K_2)+\mathbf{H}_{21}[\mathbf{w}_{21}^1, \cdots, \mathbf{w}_{21}^{K_2}]\mathbf{m}_{21}^2}_{\rm{aligned~in~complex~signal~level~(\ref{w21})}}
  +\mathbf{H}_{21}[\mathbf{v}_{21}^{K_2+1}, \cdots,  \mathbf{v}_{21}^{K_1}]\mathbf{m}_{21}^1(K_2+1:K_1)\nonumber\\
                &+&\mathbf{H}_{21}[\mathbf{u}_{21}^1, \cdots, \mathbf{u}_{21}^{K_3}]\mathbf{m}_{21}^3\nonumber\\
  &+&\underbrace{ \mathbf{H}_{22}[\mathbf{v}_{22}^1, \cdots,  \mathbf{v}_{22}^{G_2}]\mathbf{m}_{22}^1(1:G_2)+\mathbf{H}_{22}[\mathbf{w}_{22}^1, \cdots, \mathbf{w}_{22}^{G_2}]\mathbf{m}_{22}^2}_{\rm{aligned~in~complex~signal~level~(\ref{w22})}}
  +\mathbf{H}_{22}[\mathbf{v}_{22}^{G_2+1}, \cdots,  \mathbf{v}_{22}^{G_1}]\mathbf{m}_{22}^1(G_2+1:G_1)\nonumber\\
  &+&\mathbf{H}_{22}[\mathbf{u}_{22}^1, \cdots,  \mathbf{u}_{22}^{G_3}]\mathbf{m}_{22}^3
  +\underbrace{\mathbf{H}_{22}\left[\begin{array}{ccc}
                    \mathbf{u}_{12}^1& \cdots & \mathbf{u}_{12}^{J_3}
                 \end{array}\right](\mathbf{m}_{12}^3+\mathbf{m}_{11}^3)}_{\mathrm{interference~alignment~(\ref{L3c}),(\ref{J3c})}}+\mathbf{z}_2\label{Y2}
\end{eqnarray}
\end{figure*}
where $\mathbf{m}(i:j)$ denotes the the $i$th element to the $j$th element of vector $\mathbf{m}$.
According to (\ref{w21}) and (\ref{w22}), it is obvious that $\mathbf{H}_{21}\mathbf{v}_{21}^k$ is inseparable with $\mathbf{H}_{21}\mathbf{w}_{21}^k$ ($k=1,2,\cdots, K_2$)
at complex signal level, so is $\mathbf{H}_{22}\mathbf{v}_{22}^g$ and $\mathbf{H}_{22}\mathbf{w}_{22}^g$ ($g=1,2,\cdots, G_2$).
However, since all messages are real, (\ref{Y2}) can be transformed into a real system as (\ref{Y22}),
\begin{figure*}
\begin{eqnarray}
  \bar{Y}_2 &=&\bar{H}_{21}[\bar{V}_{21}^1, \cdots, \bar{V}_{21}^{K_2}]\mathbf{m}_{21}^1(1:K_2)
+\bar{H}_{21}[\bar{W}_{21}^1, \cdots, \bar{W}_{21}^{K_2}]\mathbf{m}_{21}^2+\bar{H}_{21}[\bar{V}_{21}^{K_2+1}, \cdots, \bar{V}_{21}^{K_1}]\mathbf{m}_{21}^1(K_2+1:K_1)\nonumber\\
&+&\bar{H}_{22}[\bar{V}_{22}^1, \cdots, \bar{V}_{22}^{G_2}]\mathbf{m}_{22}^1(1:G_2)
+\bar{H}_{22}[\bar{W}_{22}^1, \cdots, \bar{W}_{22}^{G_2}]\mathbf{m}_{22}^2+\bar{H}_{22}[\bar{V}_{22}^{G_2+1}, \cdots, \bar{V}_{22}^{G_1}]\mathbf{m}_{22}^1(G_2+1:G_1)~~
  \nonumber\\
                &+&\bar{H}_{21}[\bar{U}_{21}^1, \cdots, \bar{U}_{21}^{K_3}]\mathbf{m}_{21}^3
+\bar{H}_{22}[\bar{U}_{22}^1, \cdots, \bar{U}_{22}^{G_3}]\mathbf{m}_{22}^3
+\bar{H}_{22}[\bar{U}_{12}^1, \cdots, \bar{U}_{12}^{J_3}](\mathbf{m}_{12}^3+\mathbf{m}_{11}^3)+\bar{Z}_2\label{Y22}
\end{eqnarray}
\hrulefill
\end{figure*}
where
\begin{eqnarray}
  \bar{Y}_2 = \left[\begin{array}{c}
    \mathrm{Re}(\mathbf{Y}_2(1)) \\
    \mathrm{Im}(\mathbf{Y}_2(1)) \\
    \vdots \\
    \mathrm{Re}(\mathbf{Y}_2(N_2)) \\
    \mathrm{Im}(\mathbf{Y}_2(N_2))
  \end{array}\right]~,~ \bar{U}_{21}^{k_3} =\left[\begin{array}{c}
    \mathrm{Re}(\mathbf{u}_{21}^{k_3}(1)) \\
    \mathrm{Im}(\mathbf{u}_{21}^{k_3}(1)) \\
    \vdots \\
    \mathrm{Re}(\mathbf{u}_{21}^{k_3}(M_1)) \\
    \mathrm{Im}(\mathbf{u}_{21}^{k_3}(M_1))
  \end{array}\right]\nonumber\\ \bar{U}_{22}^{g_3} =\left[\begin{array}{c}
    \mathrm{Re}(\mathbf{u}_{22}^{g_3}(1)) \\
    \mathrm{Im}(\mathbf{u}_{22}^{g_3}(1)) \\
    \vdots \\
    \mathrm{Re}(\mathbf{u}_{22}^{g_3}(M_2)) \\
    \mathrm{Im}(\mathbf{u}_{22}^{g_3}(M_2))
  \end{array}\right]~,~\bar{U}_{12}^{j_3} =\left[\begin{array}{c}
    \mathrm{Re}(\mathbf{u}_{12}^{j_3}(1)) \\
    \mathrm{Im}(\mathbf{u}_{12}^{j_3}(1)) \\
    \vdots \\
    \mathrm{Re}(\mathbf{u}_{12}^{j_3}(M_2)) \\
    \mathrm{Im}(\mathbf{u}_{12}^{j_3}(M_2))
  \end{array}\right]\nonumber
\end{eqnarray}
and
\begin{eqnarray}
                 &&\bar{V}_{21}^{k} =\left[\begin{array}{c}
    \mathrm{Re}(\mathbf{v}_{21}^{k}(1)) \\
    \mathrm{Im}(\mathbf{v}_{21}^{k}(1)) \\
    \vdots \\
    \mathrm{Re}(\mathbf{v}_{21}^{k}(M_1)) \\
    \mathrm{Im}(\mathbf{v}_{21}^{k}(M_1))
  \end{array}\right]~,~ \bar{W}_{21}^{k} =\left[\begin{array}{c}
    -\mathrm{Im}(\mathbf{v}_{21}^{k}(1)) \\
    \mathrm{Re}(\mathbf{v}_{21}^{k}(1)) \\
    \vdots \\
    -\mathrm{Im}(\mathbf{v}_{21}^{k}(M_1)) \\
    \mathrm{Re}(\mathbf{v}_{21}^{k}(M_1))
  \end{array}\right]\nonumber\\&&\bar{V}_{21}^{k'} =\left[\begin{array}{c}
    \mathrm{Re}(\mathbf{v}_{21}^{k'}(1)) \\
    \mathrm{Im}(\mathbf{v}_{21}^{k'}(1)) \\
    \vdots \\
    \mathrm{Re}(\mathbf{v}_{21}^{k'}(M_1)) \\
    \mathrm{Im}(\mathbf{v}_{21}^{k'}(M_1))
  \end{array}\right]~(\mathrm{based~on~}(\ref{w21}))\label{21ind}
               \end{eqnarray}

\begin{eqnarray}
                 &&\bar{V}_{22}^{g} =\left[\begin{array}{c}
    \mathrm{Re}(\mathbf{v}_{22}^{g}(1)) \\
    \mathrm{Im}(\mathbf{v}_{22}^{g}(1)) \\
    \vdots \\
    \mathrm{Re}(\mathbf{v}_{22}^{g}(M_2)) \\
    \mathrm{Im}(\mathbf{v}_{22}^{g}(M_2))
  \end{array}\right]~,~ \bar{W}_{22}^{g} =\left[\begin{array}{c}
    -\mathrm{Im}(\mathbf{v}_{22}^{g}(1)) \\
    \mathrm{Re}(\mathbf{v}_{22}^{g}(1)) \\
    \vdots \\
    -\mathrm{Im}(\mathbf{v}_{22}^{g}(M_2)) \\
    \mathrm{Re}(\mathbf{v}_{22}^{g}(M_2))
  \end{array}\right]\nonumber\\&&\bar{V}_{22}^{g'} =\left[\begin{array}{c}
    \mathrm{Re}(\mathbf{v}_{22}^{g'}(1)) \\
    \mathrm{Im}(\mathbf{v}_{22}^{g'}(1)) \\
    \vdots \\
    \mathrm{Re}(\mathbf{v}_{22}^{g'}(M_2)) \\
    \mathrm{Im}(\mathbf{v}_{22}^{g'}(M_2))
  \end{array}\right]~(\mathrm{based~on~}(\ref{w22}))\label{22ind}
               \end{eqnarray}
for $k=1,2,\cdots, K_2$ , $k'=K_2+1,K_2+2,\cdots, K_1$, $g=1,2,\cdots, G_2$, $g'=G_2+1,G_2+2,\cdots, G_1$, $g_3=1,2,\cdots, G_3$, $k_3=1,2,\cdots, K_3$ and $j_3=1,2,\cdots, J_3$.

Next, we shall prove the independence of the received signal groups.
We first discuss the independence of the signals from transmitters $1$ and $2$, respectively. Let $\bar{V}_{rt}$, $\bar{W}_{rt}$ and $\bar{U}_{rt}$ denote the precoding matrix of $\mathbf{m}_{rt}^1$, $\mathbf{m}_{rt}^2$ and $\mathbf{m}_{rt}^3$, respectively. For example, $\bar{V}_{21}$ denotes $\left[
                                                                                   \begin{array}{ccc}
                                                                                     \bar{V}_{21}^1 & \cdots & \bar{V}_{21}^{K_1}\\
                                                                                   \end{array}\right]$, $\bar{W}_{21}$ denotes $\left[
                                                                                   \begin{array}{ccc}
                                                                                     \bar{W}_{21}^1 & \cdots & \bar{W}_{21}^{K_2}\\
                                                                                   \end{array}\right]$, and $\bar{U}_{21}$ denotes $\left[
                                                                                   \begin{array}{ccc}
                                                                                     \bar{U}_{21}^1 & \cdots & \bar{U}_{21}^{K_3}\\
                                                                                   \end{array}\right]$.

\begin{figure*}
\begin{eqnarray}\label{Y1}
  \mathbf{Y}_1&=&\mathbf{H}_{11}(\mathbf{x}_{11}^1+\mathbf{x}_{11}^2+\mathbf{x}_{11}^3)+\mathbf{H}_{12}(\mathbf{x}_{12}^1+\mathbf{x}_{12}^2+\mathbf{x}_{12}^3)
  +\underbrace{\mathbf{H}_{11}\mathbf{x}_{21}^3+\mathbf{H}_{12}\mathbf{x}_{22}^3}_{\mathrm{interference}}+\mathbf{z}_1\nonumber\\
  &=& \nonumber\underbrace{\mathbf{H}_{11}[\mathbf{v}_{11}^1, \cdots,  \mathbf{v}_{11}^{L_2}]\mathbf{m}_{11}^1(1:L_2)+\mathbf{H}_{11}[\mathbf{w}_{11}^1, \cdots, \mathbf{w}_{11}^{L_2}]\mathbf{m}_{11}^2}_{\mathrm{aligned~in~complex~signal~level~(\ref{eqn:w11align}),(\ref{L12})}}+\mathbf{H}_{11}[\mathbf{v}_{11}^{L_2+1}, \cdots,  \mathbf{v}_{11}^{L_1}]\mathbf{m}_{11}^1(L_2+1:L_1)\nonumber\\
  &+&\mathbf{H}_{11}[\mathbf{u}_{11}^1, \cdots,  \mathbf{u}_{11}^{L_3}]\mathbf{m}_{11}^3\nonumber\\
                &+&\underbrace{\mathbf{H}_{12}[\mathbf{v}_{12}^1, \cdots, \mathbf{v}_{12}^{J_2}]\mathbf{m}_{12}^1(1:J_2)+\mathbf{H}_{12}[\mathbf{w}_{12}^1,  \cdots,  \mathbf{w}_{12}^{J_2}]\mathbf{m}_{12}^2}_{\mathrm{aligned~in~complex~signal~level~(\ref{w12align}),(\ref{J2})}}+\mathbf{H}_{12}[\mathbf{v}_{12}^{J_2+1}, \cdots, \mathbf{v}_{12}^{J_1}]\mathbf{m}_{12}^1(J_2+1:J_1)\nonumber\\
                &+&\mathbf{H}_{12}\left[\begin{array}{ccc}
                    \mathbf{u}_{12}^1& \cdots & \mathbf{u}_{12}^{J_3}
                 \end{array}\right]\mathbf{m}_{12}^3
                 +\underbrace{\mathbf{H}_{11}[\mathbf{u}_{21}^1, \cdots,  \mathbf{u}_{21}^{K_3}](\mathbf{m}_{21}^3+\mathbf{m}_{22}^3)}_{\mathrm{interference~alignment~(\ref{G3}),(\ref{K3})}}+\mathbf{z}_1
\end{eqnarray}
\hrulefill
\end{figure*}

We first show that $\left[
                                   \begin{array}{ccc}
                                     \bar{V}_{21} & \bar{W}_{21} & \bar{U}_{21}\\
                                   \end{array} \right]$ has full column rank.
According to (\ref{21ind}), we can see that
$\bar{V}_{21}^k$ and $\bar{W}_{21}^k$ are independent of each other. Further, since $K_2\leq M_1-N_1<M_1$ (according to (\ref{K1})), $\left[
                                                                                   \begin{array}{cccc}
                                                                                     \bar{V}_{21}^1 & \cdots & \bar{V}_{21}^{K_2} & \bar{W}_{21} \\
                                                                                   \end{array}
                                                                                 \right]\in\mathbb{R}^{2M_1\times 2K_2}
$ has full column rank almost for sure. In addition, based on (\ref{v21}), (\ref{w21}) and (\ref{K1}), $\bar{V}_{21}^{k'}$ can be designed to guarantee that $\left[
                                                                                   \begin{array}{cc}
                                                                                     \bar{V}_{21} & \bar{W}_{21}\\
                                                                                   \end{array}\right]\in\mathbb{R}^{2M_1\times (K_1+K_2)}
$ has full column rank.
Further, (\ref{u21}) implies that $\bar{U}_{21}$
is spanning in the different space with $\left[
                                                                                   \begin{array}{cc}
                                                                                     \bar{V}_{21} & \bar{W}_{21}\\
                                                                                   \end{array}\right]$. Since $K_3\leq N_1$ and $K_1+K_2+K_3<2M_1$,
                                                                                   $\left[
                                   \begin{array}{ccc}
                                     \bar{V}_{21} & \bar{W}_{21} & \bar{U}_{21}\\
                                   \end{array} \right]\in\mathbb{R}^{2M_1\times(K_1+K_2+K_3)}$ has full column rank $K_1+K_2+K_3$ almost for sure.
Finally, the signals from transmitter 1, $\bar{H}_{21}\left[
                                   \begin{array}{ccc}
                                     \bar{V}_{21} & \bar{W}_{21} & \bar{U}_{21}\\
                                   \end{array} \right]\in\mathbb{R}^{2N_2\times (K_1+K_2+K_3)}$, will have full column rank as long as $K_1+K_2+K_3\leq 2N_2$.

Then, we consider $\bar{H}_{22}\bar{X}_{22}^1,~\bar{H}_{22}\bar{X}_{22}^2,~\bar{H}_{22}\bar{X}_{22}^3$ and $\bar{H}_{22}\bar{X}_{12}^3$ (aligned with $\bar{H}_{21}\bar{X}_{11}^3$) from transmitter 2. Their precoding matrices are $\left[
                                   \begin{array}{cccc}
                                     \bar{V}_{22} & \bar{W}_{22} & \bar{U}_{22} & \bar{U}_{12}\\
                                   \end{array} \right]$. Similar to transmitter 1, $\left[
                                   \begin{array}{ccc}
                                     \bar{V}_{22} & \bar{W}_{22} & \bar{U}_{22}\\
                                   \end{array} \right]$
can be proved to have full column rank almost for sure. For $\bar{U}_{12}$, it is designed according to (\ref{L3c}) and (\ref{p21}), which implies that it is only related to $\bar{H}_{21}$ and $\bar{H}_{22}$. Since the channels are generic and irrespective of $\left[
                                   \begin{array}{ccc}
                                     \bar{V}_{22} & \bar{W}_{22} & \bar{U}_{22}\\
                                   \end{array} \right]$,
 $\bar{U}_{12}$ can be chosen to guarantee that $\left[
                                   \begin{array}{cccc}
                                     \bar{V}_{22} & \bar{W}_{22} & \bar{U}_{22} & \bar{U}_{12}\\
                                   \end{array} \right]$ has full column rank.
                                   As we can see,  $\bar{H}_{22}\left[
                                   \begin{array}{cccc}
                                     \bar{V}_{22} & \bar{W}_{22} & \bar{U}_{22} & \bar{U}_{12}\\
                                   \end{array} \right]\in\mathbb{R}^{2N_2\times (G_1+G_2+G_3+J_3)}$ will have full column rank as long as $G_1+G_2+G_3+J_3\leq2N_2$.

 According to (\ref{Y22}),  the received signals on $R_2$ can be expressed as $\left[
                                                                              \begin{array}{cc}
                                                                                \bar{H}_{21}(\bar{V}_{21}~\bar{W}_{21}~\bar{U}_{21}) & \bar{H}_{22}(\bar{V}_{22}~\bar{W}_{22}~\bar{U}_{22}~\bar{U}_{12})\\
                                                                              \end{array}
                                                                            \right]$,
where both $\bar{H}_{21}$ and $\bar{H}_{22}$ are generic random channels. Based on above discussion, the matrix will be of full column rank as long as $K_1+K_2+K_3+G_1+G_2+G_3+J_3\leq 2N_2$. Note that the number of desired signals and interference signals on $R_2$ is $D_2=K_1+K_2+K_3+G_1+G_2+G_3$ and $J_3$, respectively.
Therefore, we have
\begin{eqnarray}\label{eqn:const}
  D_2=K_1+K_2+K_3+G_1+G_2+G_3\leq2N_2-J_3 \label{R2}
\end{eqnarray}

Next, we examine the received signals at $R_1$. The received signals in (\ref{eqn:R1})  is written in (\ref{Y1}).

Note that the structure of signal groups in (\ref{Y1}) is the same as that in (\ref{Y2}). Therefore, the independence of the signal groups can be proved in the same way as those on $R_2$.

Since $\bar{Y}_1\in \mathbb{R}^{2N_1\times1}$, the number of real dimensions on receiver $R_1$ is equal to $2N_1$. According to (\ref{Y1}), the number of
desired signals and interference signals on $R_1$ is $D_1=L_1+L_2+L_3+J_1+J_2+J_3$ and $K_3$, respectively. Therefore, the signal groups on $R_1$ will be independent of each other in real signal level as long as
\begin{eqnarray}
  D_1=L_1+L_2+L_3+J_1+J_2+J_3\leq2N_1-K_3 \label{R1}
\end{eqnarray}

Therefore, the achievable DoF can be calculated as $\frac{D_1+D_2}{2}$. Obviously, $D_1+D_2$ is maximized when the equalities of (\ref{R2}) and (\ref{R1}) both hold.

Next, we investigate the maximum achievable DoF of Case $A$, Case $B$, and Case $C$ in Section \ref{A},\ref{B}, and \ref{C}, respectively.

\section{Achievable DoF of Case $A$}\label{A}

In this section, we show the achievable DoF of our scheme in MIMO X channels for $3N_2<M_1+M_2<2N_1+N_2$ and $M_t\geq N_r$.

Maximizing the achievable DoF ($\frac{1}{2}\sum_{r=1}^{2}\sum_{t=1}^{2}Q_{rt}$) is equivalent to maximizing the number of desired signals at each receiver.

\begin{theorem} In $2\times2$ MIMO X network with $M_t$ antennas at transmitter $t$ and $N_r$ antennas at receiver $r$, when $M_1\geq M_2\geq N_1\geq N_2$ and $3N_2<M_1+M_2<2N_1+N_2$, the total achievable DoF is $\frac{M_1+M_2+N_2}{2}$ (the outer-bound). The length of each message block is shown in Table \ref{table2}.
\end{theorem}
\begin{table*}[t!]
\begin{center}\caption{Length of Message Vectors in Case $A$ ($3N_2<M_1+M_2<2N_1+N_2$)}
\renewcommand{\arraystretch}{2 }
\centering\label{table2}
\begin{tabular} {c |c |c |c |c}
\hline
$Q_{11}$  & $Q_{21}$ & $Q_{12}$ & $Q_{22}$ & Achievable DoF\\
 \hline
 $M_1-M_2+N_2$ & $M_1-M_2+N_2$ &$2(M_2-N_2)$ & $M_2-M_1+N_2$ & $\frac{M_1+M_2+N_2}{2}$\\
\hline\hline
\end{tabular}
\end{center}
\hrulefill
\end{table*}

\begin{proof} The achievable DoF is obtained by maximizing $D_1+D_2$, while satisfying the constraints of all parameters.
Therefore, it can be formulated as the following optimization problem:
\begin{align}
&\max (D_1+D_2)\nonumber\\&=\max(K_1+K_2+2K_3+G_1+G_2+L_1+L_2+2J_3+J_1+J_2)\label{opti}\\
\mathrm{st.} ~&K_1+K_2+2K_3+G_1+G_2=2N_2-J_3~(\ref{R2})\nonumber\\ &L_1+L_2+2J_3+J_1+J_2=2N_1-K_3 ~(\ref{R1})\nonumber\\
      &K_2\leq K_1\leq M_1-N_1 ~(\ref{K1})~\mathrm{and}~ G_2\leq G_1\leq M_2-N_1 ~(\ref{G1}) \nonumber\\
      &J_2\leq J_1\leq M_2-N_2 ~(\ref{J2})~\mathrm{and}~ L_2\leq L_1\leq M_1-N_2 ~(\ref{L12})\nonumber\\
      &G_3=K_3\leq N_1 ~(\ref{K3N1})~\mathrm{and}~ L_3=J_3\leq N_2 ~(\ref{J3c})\nonumber\\
      &M_1\geq M_2\geq N_1\geq N_2 ~\mathrm{and}~ 3N_2<M_1+M_2<2N_1+N_2\nonumber
      \end{align}

To solve the problem, we first maximize $K_1$, $K_2$, $G_1$, $G_2$, $J_1$ and $J_2$ by letting $K_2=K_1=M_1-N_1 ~,~ G_2= G_1= M_2-N_1 ~\mathrm{and}~ J_2= J_1= M_2-N_2$. Then, (\ref{opti}) becomes
\begin{align}
&\max (2K_3+L_1+L_2+2J_3)\label{optii}\\
\mathrm{st.} ~&2K_3+J_3=2(2N_1+N_2-M_1-M_2)~(\ref{R2}.a)\nonumber\\ &K_3+L_1+L_2+2J_3=2N_1-2M_2+2N_2 ~(\ref{R1}.a)\nonumber\\
      &L_2\leq L_1\leq M_1-N_2 ~(\ref{L12}),~G_3=K_3\leq N_1 ~(\ref{K3N1})\nonumber\\ &L_3=J_3\leq N_2 ~(\ref{J3c})\nonumber\\
      &M_1\geq M_2\geq N_1\geq N_2 ~\mathrm{and}~ 3N_2<M_1+M_2<2N_1+N_2\nonumber
      \end{align}

Taking (\ref{R1}.a) into (\ref{optii}), the optimization problem can be expressed as
\begin{align}
&\max (K_3+2N_1-2M_2+2N_2)=\max(K_3)\label{optk3}\\
\mathrm{st.} ~&2K_3+J_3=2(2N_1+N_2-M_1-M_2)~(\ref{R2}.a)\nonumber\\ &K_3+L_1+L_2+2J_3=2N_1-2M_2+2N_2 ~(\ref{R1}.a)\nonumber\\
      &L_2\leq L_1\leq M_1-N_2 ~(\ref{L12}),~G_3=K_3\leq N_1 ~(\ref{K3N1})\nonumber\\ &L_3=J_3\leq N_2 ~(\ref{J3c})\nonumber\\
      &M_1\geq M_2\geq N_1\geq N_2 ~\mathrm{and}~ 3N_2<M_1+M_2<2N_1+N_2\nonumber
      \end{align}

According to (\ref{R2}.a), we can see that maximizing $K_3$ is equivalent to minimizing $J_3$. As a result, we let $J_3=L_3=0$. Hence, $K_3=G_3=2N_1+N_2-M_1-M_2$. Note that since $M_1+M_2<2N_1+N_2$, $K_3>0$ is guaranteed. Further, since $M_1\geq M_2\geq N_1\geq N_2$, $K_3\leq N_1$ holds.

Finally, only $L_1$ and $L_2$ are left to be determined. They are constrained by
\begin{eqnarray}
  L_1+L_2=M_1-M_2+N_2 ~\mathrm{and}~L_2\leq L_1\leq M_1-N_2
\end{eqnarray}

Accordingly, we choose $L_1$ and $L_2$ as follows.
\begin{eqnarray}
 L_1=\lceil\frac{M_1-M_2+N_2}{2}\rceil ~\mathrm{and}~ L_2=\lfloor\frac{M_1-M_2+N_2}{2}\rfloor\label{L1+}
\end{eqnarray}

Obviously, $L_1+L_2=M_1-M_2+N_2$. Next, we show that (\ref{L12}) is also satisfied, i.e., $L_2\leq L_1\leq M_1-N_2$. Since $3N_2<M_1+M_2$, it is easy to get that $\frac{M_1-M_2+N_2}{2}<M_1-N_2$, which leads to
\begin{eqnarray}
 \lceil\frac{M_1-M_2+N_2}{2}\rceil\leq \lceil M_1-N_2\rceil\label{M1N2}
\end{eqnarray}

Since $M_1-N_2$ is always an integer, $\lceil M_1-N_2\rceil=M_1-N_2$. Hence, from (\ref{L1+}) and (\ref{M1N2}) we can get that
\begin{eqnarray}
 L_2\leq L_1=\lceil\frac{M_1-M_2+N_2}{2}\rceil\leq M_1-N_2.
\end{eqnarray}

Therefore, (\ref{L12}) is satisfied.
As the length of all message groups have been determined, finally we have
\begin{eqnarray}
  Q_{11}&=&L_1+L_2=M_1-M_2+N_2\nonumber\\
  Q_{21}&=&K_1+K_2+K_3=M_1-M_2+N_2\nonumber\\
  Q_{12}&=&J_1+J_2=2(M_2-N_2)\nonumber\\
  Q_{22}&=&G_1+G_2+G_3=M_2-M_1+N_2
\end{eqnarray}

Note that since $M_1\geq M_2$, $Q_{11}=Q_{21}>0$. Since $3N_2<2N_1+N_2$, we have $M_2\geq N_1>N_2$ and $Q_{12}>0$. Since $2N_1+N_2>M_1+M_2\geq M_1+N_1$, we can get $N_1+N_2>M_1$. Hence, we have $M_2+N_2\geq N_1+N_2>M_1$ and $Q_{22}>0$.

Finally, the DoF can be calculated as $\frac{Q_{11}+Q_{21}+Q_{12}+Q_{22}}{2}=\frac{M_1+M_2+N_2}{2}$, which equals the outer bound.
\end{proof}

\begin{example} Two examples of ($M_1,~M_2,~N_1,~N_2$) are given, which are $(2,2,2,1)$ and $(7,6,5,4)$, respectively. For $(2,2,2,1)$, we can get $L_1=1$ and $L_2=L_3=0$; $K_1=K_2=0$ and $K_3=1$; $J_1=J_2=1$ and $J_3=0$; $G_1=G_2=0$ and $G_3=1$. Five signals are transmitted, achieving DoF of $\frac{5}{2}$. For $(7,6,5,4)$, we have $L_1=3$, $L_2=2$ and $L_3=0$; $K_1=K_2=2$ and $K_3=1$; $J_1=J_2=2$ and $J_3=0$; $G_1=G_2=G_3=1$. Totally $17$ signals are transmitted, achieving DoF of $\frac{17}{2}$. The outer-bound DoF is achieved in both examples.
\end{example}

\begin{remark} The results in \emph{Theorem 1} can be explained in a general and straightforward way as follows.


The network can be viewed as three concatenated sub-networks, as shown in Fig. \ref{a}.  In sub-network $1$, link $T_2$--$R_1$ only contains messages intended to $R_1$. In sub-network $2$, link $T_1$--$R_2$ and link $T_2$--$R_2$ both contain messages intended to $R_2$. In sub-network $3$, link $T_1$--$R_1$ contains messages intended to $R_1$.


\begin{figure}[t!]
  \begin{center}
        \includegraphics[width=1\columnwidth]{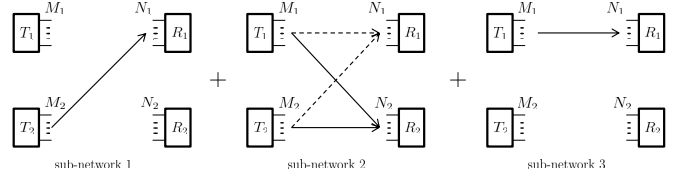}
        \caption{The three concatenated sub-networks (Case $A$)}
        \label{a}
    \end{center}
\end{figure}

In sub-network $1$, there are equivalently $2N_1$ real dimensions for link $T_2$--$R_1$, and $2(M_2-N_2)$ of them are interference free for $R_2$. To maximize $Q_{12}$ while avoid interfering $R_2$, $T_2$ transmits $2(M_2-N_2)$ messages via those $2(M_2-N_2)$ interference free dimensions, i.e, $Q_{12}= 2(M_2-N_2)$. (Note that in case $A$, $M_2$ is always larger than $N_2$.)

In sub-network $2$, there are now $2(N_1+N_2-M_2)$ and $2N_2$ dimensions on $R_1$ and $R_2$, respectively. Links $T_1$--$R_2$ and $T_2$--$R_2$ have $2(M_1-N_1)$ and $2(M_2-N_1)$ real dimensions that are interference free for $R_1$, respectively. These dimensions will be chosen at first by their corresponding transmitters, and will occupy totally $2(M_1-N_1)+2(M_2-N_1)=2M_1+2M_2-4N_1$ real dimensions on $R_2$. As we can see, there are still $2(2N_1+N_2-M_1-M_2)$ dimensions available on $R_2$ (note that $2N_1+N_2>M_1+M_2$ in Case $A$), which means it can still accommodate $2(2N_1+N_2-M_1-M_2)$ more messages. Hence, $T_1$ and $T_2$ can use $2N_1+N_2-M_1-M_2$ more dimensions to transmit messages to $R_2$. Note that these signals can be aligned one-to-one at $R_1$, and thereby generating totally $2N_1+N_2-M_1-M_2$ interference dimensions at $R_1$. Therefore, we have $Q_{21}=2(M_1-N_1)+(2N_1+N_2-M_1-M_2)=M_1-M_2+N_2$ and $Q_{22}=2(M_2-N_1)+(2N_1+N_2-M_1-M_2)=M_2-M_1+N_2$.

Then, in sub-network $3$ there are now only $2(N_1+N_2-M_2)-(2N_1+N_2-M_1-M_2)=M_1-M_2+N_2$ dimensions left on $R_1$ and no dimensions left on $R_2$. It implies that at most $M_1-M_2+N_2$ messages can be transmitted to $R_1$ through link $T_1$--$R_1$, but no interference can be caused on $R_2$. Note that the number of dimensions that are interference free for $R_2$ on link $T_1$--$R_1$ is $2(M_1-N_2)$, and note that $2(M_1-N_2)>M_1-M_2+N_2$ (because $M_1+M_2>3N_2$), we can always find $M_1-M_2+N_2$ real dimensions to transmit $M_1-M_2+N_2$ messages through link $T_1$--$R_1$ without generating any interference to $R_2$. Therefore, $Q_{11}=M_1-M_2+N_2$.
\end{remark}

\section {Achievable DoF of Case $B$}\label{B}

In this section, we show the achievable DoF of our scheme in MIMO X channels for $M_1+M_2\geq2N_1+N_2$ and $M_t\geq N_r$.

\begin{theorem} In $2\times2$ MIMO X network with $M_t$ antennas at transmitter $t$ and $N_r$ antennas at receiver $r$, when $M_1\geq M_2\geq N_1\geq N_2$ and $M_1+M_2\geq2N_1+N_2$, the achievable DoF equals\\

~~~~~~~~~~~$\left\{\begin{array}{ll}
      N_1+N_2-1 & \mathrm{if}~ M_1\geq M_2=N_1=N_2\\
     N_1+N_2-\frac{1}{2} & \mathrm{if}~ M_1\geq M_2=N_1>N_2\\
     N_1+N_2  & \mathrm{if} ~M_1\geq M_2>N_1\geq N_2
   \end{array}\right.$\\

The length of each message block in different subcases is shown in Table \ref{table3}.\end{theorem}

\begin{table*}
[t!]
\begin{center}\caption{Length of Message Vectors in Case $B$ ($M_1+M_2\geq 2N_1+N_2$)}
\renewcommand{\arraystretch}{2 }
\centering\label{table3}
\begin{tabular} {c |c |c |c| c| c}
\hline \hline
 & $Q_{11}$  & $Q_{21}$ & $Q_{12}$ & $Q_{22}$ & Achievable DoF\\
\hline
(1)$M_1\geq M_2=N_1=N_2$ &$ 2N_1-2 $ & $ 2N_2-2 $& $1$ & $1$ & $N_1+N_2-1$\\
\hline
(2)$M_1\geq M_2=N_1>N_2 $ &  $2N_2-1$ &$ 2N_2-1$ &$2(M_2-N_2)$ & $1$ & $N_1+N_2-\frac{1}{2}$\\
\hline

(3)$M_1\geq M_2>N_1\geq N_2$ & $2N_1-Q_{12}$ & $2N_2-Q_{22}$ & $\min\{2(M_2-N_2)~,N_1\}$ &$\min\{2(M_2-N_1)~,N_2\}$ & $N_1+N_2$\\
\hline\hline
\end{tabular}
\end{center}
\hrulefill
\end{table*}

We divide Case $B$ into three subcases as shown in Table \ref{table3}. The achievable DoF of each subcase is investigated one by one as follows.

\subsection{When $M_1\geq M_2=N_1=N_2$ ((1) of Case $B$)}

\begin{proof} First, since $M_1+M_2\geq 2N_1+N_2$ and $M_2=N_1=N_2$, we can exclude $M_1=M_2=N_1=N_2$ from this subcase.

According to (\ref{G1}) and (\ref{J2}), we can get $G_1=G_2=J_1=J_2=0$. To ensure $Q_{12}=J_1+J_2+J_3>0$ and $Q_{22}=G_1+G_2+G_3>0$, $J_3\geq1$ and $G_3\geq1$ must be added as constraints.

Also note that in this case (\ref{R2}) and (\ref{R1}) can be expressed as
\begin{eqnarray}
  Q_{21}+Q_{22}+Q_{12} = 2N_2\nonumber \\
  Q_{11}+Q_{12}+Q_{22} = 2N_1 \nonumber
\end{eqnarray}
which is equivalent to
\begin{eqnarray}
  Q_{21}+Q_{22}+Q_{12}+Q_{11}+Q_{12}+Q_{22} = 2N_1+2N_2 \nonumber
\end{eqnarray}

Since $Q_{rt}\geq1,~(r,t=1,2)$ and $N_1=N_2$, we can get $N_1=N_2\geq \frac{3}{2}$. It implies that $Q_{rt}\geq1,~(r,t=1,2)$ can only be achieved when $N_1=N_2>1$. Therefore, in this case we only need to consider $N_1>1$.

To maximize the achievable DoF, the optimization problem can be expressed as
\begin{align}
&\max (D_1+D_2)\nonumber\\&=\max(K_1+K_2+2K_3+L_1+L_2+2J_3)\label{opti2}\\
\mathrm{st.} ~&K_1+K_2+2K_3=2N_2-J_3~(\ref{R2}.b)\nonumber\\ &L_1+L_2+2J_3=2N_1-K_3 ~(\ref{R1}.b)\nonumber\\
      &K_2\leq K_1\leq M_1-N_1 ~(\ref{K1}),~L_2\leq L_1\leq M_1-N_2 ~(\ref{L12})  \nonumber\\
      &1\leq J_3=L_3\leq N_2~\mathrm{and}~ 1\leq G_3=K_3\leq N_1\nonumber\\
      &M_1>M_2= N_1= N_2 ~\mathrm{and}~ M_1+M_2\geq2N_1+N_2\nonumber
      \end{align}

By taking (\ref{R2}.$b$) and (\ref{R1}.$b$) into (\ref{opti2}), the optimization objective becomes
\begin{eqnarray}
  \max (2N_2+2N_1-J_3-K_3)= \min(J_3+K_3)\nonumber
\end{eqnarray}

Therefore, we choose $J_3=L_3=K_3=G_3=1$. Then, (\ref{R2}.$b$) and (\ref{R1}.$b$) can be written as
\begin{eqnarray}
  K_1+K_2 =2N_2-3~\mathrm{and}~K_2\leq K_1\leq M_1-N_1~(\ref{K1})\nonumber \\
  L_1+L_2=2N_1-3~\mathrm{and}~L_2\leq L_1\leq M_1-N_2~(\ref{L12})\nonumber
\end{eqnarray}

To satisfy the above constraints, we can choose
\begin{eqnarray}
  K_1=N_2-1~\mathrm{and}~K_2=N_2-2\nonumber \\
  L_1=N_1-1~\mathrm{and}~L_2=N_1-2\nonumber
\end{eqnarray}

Since $ M_1+M_2\geq2N_1+N_2$ and $M_2=N_1=N_2$, we can get $M_1\geq N_1+N_2$. Then, it is easy to prove that $K_1=N_2-1\leq M_1-N_1$ and $L_1=N_1-1\leq M_1-N_2$. Therefore, (\ref{K1}) and (\ref{L12}) are satisfied.

Finally, the length of each message block can be calculated as
\begin{eqnarray}
  Q_{11}&=&L_1+L_2+L_3=2N_1-2\nonumber\\
  Q_{21}&=&K_1+K_2+K_3=2N_2-2\nonumber\\
  Q_{12}&=&J_3=1\nonumber\\
  Q_{22}&=&G_3=1
\end{eqnarray}
as shown in (1) of Table \ref{table3}. As we can see, when $N_1=N_2>1$, $Q_{rt}\geq1,~(r,t=1,2)$. Hence, the DoF equals $\frac{Q_{11}+Q_{21}+Q_{12}+Q_{22}}{2}=N_1+N_2-1$. Note that the outer-bound for this case is $N_1+N_2$ \cite{MIMO}.
\end{proof}

\begin{example} One example for this case is $(6,3,3,3)$. Accordingly, we can get $L_1=2$ and $L_2=L_3=1$; $K_1=2$ and $K_2=K_3=1$; $G_1=G_2=0$ and $G_3=1$; $J_1=J_2=0$ and $J_3=1$. Totally $10$ signals are transmitted, achieving DoF of $5$. The outer bound is $N_1+N_2=6$.
\end{example}

\begin{remark} The network in subcase (1) can be divided into two concatenated sub-networks as shown in Fig. \ref{pb}. In sub-network 1, link $T_2$--$R_1$ and link $T_2$--$R_2$ contain messages intended to $R_1$ and $R_2$, respectively. Since $M_2=N_2$, link $T_2$--$R_1$ does not have any dimension that is interference free for $R_2$. To minimize the interference, it only transmits one message that occupies one interference dimension on $R_2$. Similarly, since $M_2=N_1$, link $T_2$--$R_2$ does not have any dimension that is interference free for $R_1$, it only transmits one message that occupies one interference dimension on $R_1$. As a result, $Q_{12}=Q_{22}=1$.

\begin{figure}[t!]
  \begin{center}
        \includegraphics[width=0.7\columnwidth]{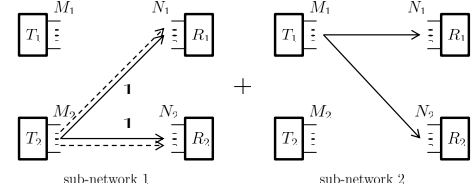}
        \caption{The two concatenated sub-networks ((1) of Case $B$)}
        \label{pb}
    \end{center}
\end{figure}
Now, there are $2N_1-1$ and $2N_2-1$ dimensions left on $R_1$ and $R_2$, respectively, which indicates that at most $2N_1-1$ messages can be transmitted to $R_1$ or $R_2$. In sub-network 2, link $T_1$--$R_1$ and link $T_1$--$R_2$ contain messages intended to $R_1$ and $R_2$, respectively. Note that there are $2(M_1-N_2)$ real dimensions in link $T_1$--$R_1$ that are interference free for $R_2$ and $2(M_1-N_2)\geq2N_1-1$ ($M_1\geq2N_1=2N_2$). Therefore, totally $2N_1-1$ messages can be transmitted in link $T_1$--$R_1$ via $2N_1-1$ dimensions that do not cause interference at $R_2$. Similarly, since there are $2(M_1-N_1)$ real dimensions in link $T_1$--$R_2$ that are interference free for $R_1$ and $2(M_1-N_1)\geq2N_2-1$ ($M_1\geq2N_1=2N_2$), totally $2N_2-1$ messages can be transmitted in link $T_1$--$R_2$ via $2N_2-1$ dimensions that do not cause interference at $R_1$. As a result, $Q_{11}=Q_{21}=2N_1-1=2N_2-1$.

Note that there is an alternative setup for links $T_1$--$R_1$ and $T_1$--$R_2$ in subnetwork 2. For link $T_1$--$R_1$, among the $2N_1-1$ messages to be sent, $2N_1-2$ of them are sent via $2N_1-2$ dimensions that are interference free for $R_2$. The last message is sent through a dimension that will cause interference to $R_2$, but the interference is aligned with that caused by link $T_2$--$R_1$. For link $T_1$--$R_2$, similarly, $2N_2-2$ messages can be sent via $2N_1-2$ dimensions that are interference free for $R_1$, while the last message is aligned with the interference caused by link $T_2$--$R_2$ on $R_1$.
This setup well matches our proposed signal design, while the results remain the same. It implies that there are multiple ways to design the transmitted signals to obtain the same achievable DoF.

Also note that if we let $T_2$ remain silent, $T_1$ can transmit $2N_1$ and $2N_2$ messages to $R_1$ and $R_2$, respectively, without generating any interference ($2(M_1-N_1)=2(M_1-N_2)\geq2N_1=2N_2$). In that case the optimal DoF is ($N_1+N_2$), but it is not an X network but a broadcast network.\end{remark}

\subsection{When $M_1\geq M_2=N_1>N_2$ ((2) of Case $B$)}

\begin{proof} First, since $M_1+M_2\geq 2N_1+N_2$ and $M_2=N_1$, we can exclude $M_1=M_2=N_1>N_2$ from this case and only focus on $M_1>M_2=N_1>N_2$.

Since $M_2=N_1$, we have $G_1=G_2=0$ (according to (\ref{G1})). Consequently, $K_3=G_3=Q_{22}\geq1$ must be added as one constraint of the optimization problem. Specifically, it can be written as
\begin{align}
&\max (D_1+D_2=K_1+K_2+2K_3+L_1+L_2+2J_3+J_1+J_2)\label{opti2'}\nonumber\\
\mathrm{st.} ~&K_1+K_2+2K_3=2N_2-J_3~(\ref{R2}.b)\nonumber\\ &L_1+L_2+2J_3+J_1+J_2=2N_1-K_3 ~(\ref{R1}.c)\nonumber\\
      &K_2\leq K_1\leq M_1-N_1 ~(\ref{K1}),~J_2\leq J_1\leq M_2-N_2 ~(\ref{J2})\nonumber\\ &L_2\leq L_1\leq M_1-N_2 ~(\ref{L12})\nonumber\\
      &1\leq G_3=K_3\leq N_1~\mathrm{and}~J_3=L_3\leq N_2\nonumber\\
      &M_1>M_2=N_1>N_2 ~\mathrm{and}~ M_1+M_2\geq 2N_1+N_2\nonumber
      \end{align}

Similar to subcase (1) of Case $B$, the optimization objective can be expressed as
\begin{eqnarray}
  \max (2N_2+2N_1-J_3-K_3)= \min(J_3+K_3)\nonumber
\end{eqnarray}

Accordingly, $J_3$ and $K_3$ can be chosen as $J_3=L_3=0$ and $K_3=G_3=1$. Then, based on (\ref{R2}.$b$) and (\ref{R1}.$c$) we have
\begin{eqnarray}
  &&K_1+K_2 =2N_2-2 ~\mathrm{and}~K_2\leq K_1\leq M_1-N_1 \nonumber\\
  &&L_1+L_2+J_1+J_2=2N_1-1\label{LJ}\nonumber\\
  &&J_2\leq J_1\leq M_2-N_2~\mathrm{and}~ L_2\leq L_1\leq M_1-N_2\nonumber
\end{eqnarray}

First, we choose
\begin{eqnarray}
  K_1=K_2=N_2-1\nonumber
\end{eqnarray}

Since $M_1+M_2\geq 2N_1+N_2$ and $M_2=N_1$, we have $M_1\geq N_1+N_2$. It is clearly that $K_1=N_2-1\leq M_1-N_1$, so (\ref{K1}) is satisfied.

Then, we choose $J_1=J_2=M_2-N_2$, so (\ref{J2}) is satisfied. Consequently, we get $L_1+L_2=2N_2-1$. We choose $L_1=N_2$ and $L_2=N_2-1$. It can be seen that $L_1=N_2\leq M_1-N_2$ because $M_1\geq N_1+N_2\geq 2N_2$, so (\ref{L12}) is satisfied.

In summary, the length of each message block can be calculated as
\begin{eqnarray}
  Q_{11}&=&L_1+L_2=2N_2-1\nonumber\\
  Q_{21}&=&K_1+K_2+K_3=2N_2-1\nonumber\\
  Q_{12}&=&J_1+J_2=2(M_2-N_2)\nonumber\\
  Q_{22}&=&G_3=1
\end{eqnarray}

The achievable DoF equals $\frac{Q_{11}+Q_{21}+Q_{12}+Q_{22}}{2}=N_1+N_2-\frac{1}{2}$. The outer-bound for this case is $N_1+N_2$ \cite{MIMO}.
\end{proof}

\begin{example} An example of this case is $(8,4,4,3)$. We can get $L_1=3$, $L_2=2$ and $L_3=0$; $K_1=K_2=2$ and $K_3=1$; $J_1=J_2=1$ and $J_3=0$; $G_1=G_2=0$ and $G_3=1$. Totally $13$ signals are transmitted, achieving DoF of $6.5$. The outer bound is $N_1+N_2=7$.
\end{example}

\begin{remark} The results for (2) of Case $B$ can be justified intuitively by two concatenated sub-networks as shown in Fig. \ref{b2}.
\begin{figure}[t!]
  \begin{center}
        \includegraphics[width=0.7\columnwidth]{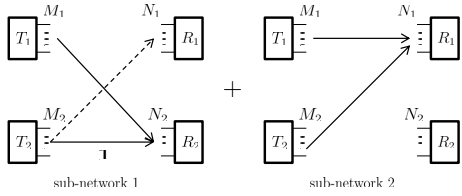}
        \caption{The two concatenated sub-networks ((2) of Case $B$)}
        \label{b2}
    \end{center}
\end{figure}
In sub-network 1, link $T_1$--$R_2$ and link $T_2$--$R_2$ both contain messages intended to $R_2$. Since $M_2=N_1$, link $T_2$--$R_2$ does not have any dimension that interference free for $R_1$. Therefore, it only transmits one message to $R_2$, while occupying one interference dimension on $R_1$. As a consequence, there are $2N_2-1$ real dimensions left on $R_2$, which means at most $2N_2-1$ messages can be transmitted through link $T_1$--$R_2$. Note that
there are $2(M_1-N_1)$ dimensions in link $T_1$--$R_2$ that are interference free for $R_1$. Since $M_1+M_2\geq2N_1+N_2$ and $M_1\geq N_1+N_2$, we can get $2(M_1-N_1)\geq2N_2-1$. Therefore, all $2N_2-1$ messages can be transmitted through link $T_1$--$R_2$ without generating any interference to $R_1$. Therefore, we have $Q_{22}=1$ and $Q_{21}=2N_2-1$.

In sub-network 2, link $T_1$--$R_1$ and link $T_2$--$R_1$ both contain messages intended to $R_1$. Note that there are $2N_1-1$ and zero dimensions left on $R_1$ and $R_2$, respectively, which means at most $2N_1-1$ messages can be transmitted to $R_1$ and no interference can be caused on $R_2$. In link $T_2$--$R_1$, there are $2(M_2-N_2)$ dimensions that are interference free for $R_2$, which are all used for transmitting $2(M_2-N_2)$ messages to $R_1$. Then, there are only $2N_2-1$ dimensions left on $R_1$, which means at most $2N_2-1$ messages can be transmitted via link $T_1$--$R_1$. Since there are $2(M_1-N_2)$ dimensions that are interference free for $R_2$ in link $T_1$--$R_1$, and $2(M_1-N_1)\geq2N_2-1$, all $2N_2-1$ messages can be transmitted without generating any interference to $R_2$. As a consequence, $Q_{12}=2(M_2-N_2)$ and $Q_{11}=2N_2-1$.

Similar to subcase (1), if we let link $T_2$--$R_2$ remain silent and link $T_1$--$N_2$  transmit $2N_2$ messages to $R_2$ with $2N_2$ dimensions that are interference free for $R_1$, then no interference will be caused on any receiver and the outer-bound DoF can be achieved. However, it is not an X network but a Z network.
\end{remark}

\subsection{When $M_1\geq M_2>N_1\geq N_2$ ((3) of Case $B$)}

\begin{proof} The optimization problem is formulated as follows.
\begin{align}
&\max (2N_2+2N_1-K_3-J_3)= \min (K_3+J_3)\\
\mathrm{st.} ~&K_1+K_2+2K_3+G_1+G_2=2N_2-J_3~(\ref{R2}.c)\nonumber\\ &L_1+L_2+2J_3+J_1+J_2=2N_1-K_3 ~(\ref{R1}.c)\nonumber\\
      &K_2\leq K_1\leq M_1-N_1 ~(\ref{K1}),~ G_2\leq G_1\leq M_2-N_1 ~(\ref{G1}) \nonumber\\
      &J_2\leq J_1\leq M_2-N_2 ~(\ref{J2}),~ L_2\leq L_1\leq M_1-N_2 ~(\ref{L12})\nonumber\\
      & G_3=K_3\leq N_1~\mathrm{and}~J_3=L_3\leq N_2\nonumber\\
      &M_1\geq M_2>N_1\geq N_2 ~\mathrm{and}~ M_1+M_2\geq2N_1+N_2\nonumber
      \end{align}

To minimize $K_3+J_3$, we can let $J_3=L_3=K_3=G_3=0$. Then, (\ref{R2}.$c$) and (\ref{R1}.$c$) become
\begin{eqnarray}
  K_1+K_2+G_1+G_2=2N_2\label{K+} \\
  L_1+L_2+J_1+J_2=2N_1 \label{L+}
\end{eqnarray}
with the constraints of (\ref{K1}), (\ref{G1}), (\ref{L12}) and (\ref{J2}).

We first determine $K_1$, $K_2$, $G_1$, and $G_2$ based on (\ref{K+}), (\ref{K1}) and (\ref{G1}).

When $2(M_2-N_1)\geq N_2$, we let $G_1=K_1=\lceil\frac{N_2}{2}\rceil$ and $G_2=K_2=\lfloor\frac{N_2}{2}\rfloor$. We can see that $\lceil\frac{N_2}{2}\rceil\leq\lceil M_2-N_1\rceil=(M_2-N_1)\leq M_1-N_1$, which means (\ref{K1}) and (\ref{G1}) are both satisfied. Further, $Q_{21}=Q_{22}=N_2>0$.

When $2(M_2-N_1)< N_2$, we have $\lceil\frac{N_2}{2}\rceil > M_2-N_1$. If we still let $G_1=\lceil\frac{N_2}{2}\rceil$, then (\ref{G1}) will not be satisfied. As a consequence, we choose $G_1=G_2=M_2-N_1$ and $K_1=K_2=N_1+N_2-M_2$. Since $2N_1+N_2\leq M_1+M_2$, it can be proved that $K_1=N_1+N_2-M_2\leq M_1-N_1$. Hence, both (\ref{K1}) and (\ref{G1}) are satisfied. In addition, $Q_{22}=G_1+G_2=2(M_2-N_1)>0$ and $Q_{21}=K_1+K_2=2N_2-2(M_2-N_1)>0$ can be guaranteed.

Next, we determine $L_1$, $L_2$, $J_1$ and $J_2$ based on (\ref{L+}), (\ref{J2}) and (\ref{L12}).

When $2(M_2-N_2)\geq N_1$, we choose $L_1=J_1=\lceil\frac{N_1}{2}\rceil$ and $L_2=J_2=\lfloor\frac{N_1}{2}\rfloor$. We can prove that $\lceil\frac{N_1}{2}\rceil\leq\lceil M_2-N_2\rceil=(M_2-N_2)<(M_1-N_2)$ and $Q_{11}=Q_{12}=N_1>0$.

When $2(M_2-N_2)<N_1$, we can choose $J_1=J_2=M_2-N_2$ and $L_1=L_2=N_1+N_2-M_2$. Since $L_1=N_1+N_2-M_2\leq M_1-N_2$, (\ref{L12}) and (\ref{J2}) are satisfied. In addition, $Q_{11}=L_1+L_2=2N_1-2(M_2-N_2)>0$ and $Q_{12}=J_1+J-2=2(M_2-N_2)>0$ can be guaranteed.

The length of each message block can be calculated as
\begin{eqnarray}
  Q_{12}&=&J_1+J_2=\min\{2(M_2-N_2), N_1\}\nonumber\\
  Q_{11}&=&L_1+L_2=2N_1-Q_{12}=\max\{2N_1+2N_2-2M_2, N_1\}\nonumber\\
  Q_{22}&=&G_1+G_2=\min\{2(M_2-N_1), N_2\}\nonumber\\
  Q_{21}&=&K_1+K_2=2N_2-Q_{22}=\max\{2N_1+2N_2-2M_2, N_2\}\nonumber
\end{eqnarray}

The achievable DoF can be calculated as $\frac{Q_{11}+Q_{12}+Q_{21}+Q_{22}}{2}=N_1+N_2$, which is equal to the outer-bound \cite{MIMO}.\end{proof}

\begin{example} Two examples are given for this case, which are $(4,4,3,2)$ and $(8,7,5,5)$. For $(4,4,3,2)$, we have $G_1=G_2=1$ and $G_3=0$; $K_1=K_2=1$ and $K_3=0$; $L_1=2$, $L_2=1$ and $L_3=0$; $J_1=2$, $J_2=1$ and $J_3=0$. Ten signals are sent, achieving DoF of $5$. For $(8,7,5,5)$, we have $G_1=G_2=2$ and $G_3=0$; $K_1=K_2=3$ and $K_3=0$; $L_1=L_2=3$ and $L_3=0$; $J_1=J_2=2$ and $J_3=0$. Twenty signals are sent, achieving DoF of $10$. The outer-bound DoF is achieved in both examples.\end{example}

\begin{remark} The intuitive explanation of this subcase can be referred to Fig. \ref{p2}.

\begin{figure}[t!]
  \begin{center}
        \includegraphics[width=0.7\columnwidth]{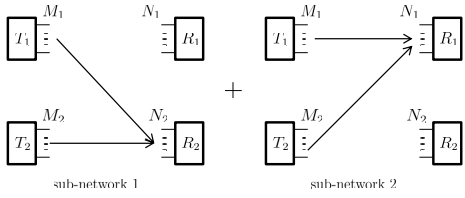}
        \caption{The two concatenated sub-networks ((3) of Case $B$)}
        \label{p2}
    \end{center}
\end{figure}

In sub-network 1, link $T_1$--$R_2$ and link $T_2$--$R_2$ both contain messages intended to $R_2$. Note that there are $2(M_1-N_1)$ and $2(M_2-N_1)$ dimensions that are interference free for $R_1$ in link $T_1$--$R_2$ and link $T_2$--$R_2$, respectively. Since $2(M_1-N_1)+2(M_2-N_1)>2N_2$, totally $2N_2$ messages can be transmitted to $R_2$ via the two links without generating any interference to $R_1$.
The similar argument can be made in sub network 2, totally $2N_1$ messages can be transmitted to $R_1$ via link $T_1$--$R_1$ and link $T_2$--$R_1$ without interfering $R_2$.
Therefore, totally $2N_1+2N_2$ messages can be transmitted.
\end{remark}

\section{Achievable DoF of Case $C$ }\label{C}

In this section, we show the achievable DoF of our scheme in MIMO X channels for $M_1+M_2\leq3N_2$ and $M_t\geq N_r$.

\begin{theorem} In $2\times2$ MIMO X network with $M_t$ antennas at transmitter $t$ and $N_r$ antennas at receiver $r$, when $M_1\geq M_2\geq N_1\geq N_2$ and $M_1+M_2\leq3N_2$, the achievable DoF equals
\begin{eqnarray}
  M_1+M_2-N_2+\frac{1}{2}\lfloor\frac{2x}{3}\rfloor =~~~~~~~~~~~~~~~~~~~~~~ \nonumber\\\left\{\begin{array}{ll}
      \frac{2(M_1+M_2)}{3} & \mathrm{when}~ x\mod 3=0\\
     \frac{2(M_1+M_2)-1}{3} & \mathrm{when}~ x\mod 3=1\\
     \frac{4(M_1+M_2)-1}{6}  &\mathrm{when}~ x\mod 3=2
   \end{array}\right.\nonumber
\end{eqnarray}
where $x=3N_2-(M_1+M_2)$.
The length of each message block is shown in Table \ref{table4}.\end{theorem}

\begin{table*}
[t!]
\begin{center}\caption{Length of Message Vectors in Case $C$ ($3N_2\geq M_1+M_2$)}
\renewcommand{\arraystretch}{2 }
\centering\label{table4}
\begin{tabular} {c |c |c |c| c}
\hline \hline
 $Q_{11}$  & $Q_{21}$ & $Q_{12}$ & $Q_{22}$ & achievable Dof\\
\hline
  $2(M_1-N_2)+\lfloor\frac{2x}{3}\rfloor$ &$2(M_1-N_2)+\lfloor\frac{2x}{3}\rfloor$ &$M_2+N_2-M_1-\lceil\frac{x-1}{3}\rceil$ & $M_2+N_2-M_1-\lfloor\frac{x}{3}\rfloor$ & $M_1+M_2-N_2+\frac{1}{2}\lfloor\frac{2x}{3}\rfloor$\\
\hline\hline
\end{tabular}
\end{center}
\hrulefill
\end{table*}

\begin{proof} In this case, there is a slight difference in the design of precoders. Specifically, we let $G_3\geq K_3$ and $J_3\geq L_3$ instead of $K_3= G_3$ and $J_3=L_3$, but (\ref{G3}) and (\ref{L3c}) still hold. 

In addition, the equalities of (\ref{R2}) and (\ref{R1}) may not always hold. Therefore, the optimization problem can be expressed as
\begin{align}
&\max (K_1+K_2+K_3+G_1+G_2+G_3\nonumber\\&~~~~~~+L_1+L_2+L_3+J_1+J_2+J_3)\label{opti3}\\
\mathrm{st.} ~&K_1+K_2+K_3+G_1+G_2+G_3\leq2N_2-J_3~(\ref{R2})\nonumber\\ &L_1+L_2++L_3+J_1+J_2+J_3\leq2N_1-G_3 ~(\ref{R1})\nonumber\\
      &K_2\leq K_1\leq M_1-N_1 ~(\ref{K1}),~ G_2\leq G_1\leq M_2-N_1 ~(\ref{G1}) \nonumber\\
      &J_2\leq J_1\leq M_2-N_2 ~(\ref{J2}),~ L_2\leq L_1\leq M_1-N_2 ~(\ref{L12})\nonumber\\
      &K_3\leq G_3\leq N_1 ~(\ref{K3N1})~\mathrm{and}~ L_3\leq J_3\leq N_2 ~(\ref{J3c})\nonumber\\
      &M_1\geq M_2\geq N_1\geq N_2 ~\mathrm{and}~ M_1+M_2\leq3N_2\nonumber
      \end{align}

First, we maximize $K_1$, $K_2$, $G_1$, $G_2$, $L_1$, $L_2$, $J_1$ and $J_2$, i.e.,
\begin{eqnarray}
  K_1=K_2 = M_1-N_1,~G_1=G_2 = M_2-N_1\nonumber\\ L_1=L_2=M_1-N_2,~ J_1=J_2= M_2-N_2\label{KLJ}
\end{eqnarray}

Then, the optimization problem can be written as
\begin{align}
&\max (4M_1+4M_2-4N_1-4N_2+K_3+G_3+L_3+J_3)\nonumber\\&=\max (K_3+L_3+G_3+J_3)\label{opti3'}\\
\mathrm{st.} ~&2(M_1+M_2-2N_1)+K_3+G_3\leq2N_2-J_3~(\ref{R2})\nonumber\\ &2(M_1+M_2-2N_2)+J_3+L_3\leq2N_1-G_3 ~(\ref{R1})\nonumber\\
    &K_3\leq G_3\leq N_1 ~(\ref{K3N1}),~ L_3\leq J_3\leq N_2 ~(\ref{J3c})\nonumber\\
      &M_1\geq M_2\geq N_1\geq N_2 ~\mathrm{and}~ M_1+M_2\leq3N_2\nonumber
      \end{align}

To maximize $G_3+J_3+K_3+L_3$, we first maximize $G_3+J_3$ by letting the equality of (\ref{R1}) hold, i.e,
\begin{eqnarray}
 G_3+J_3=2(N_1+2N_2-M_1-M_2)-L_3\label{cccc}
\end{eqnarray}

Substituting (\ref{cccc}) into (\ref{R2}), we have
\begin{eqnarray}
  K_3 \leq 2(N_1-N_2)+L_3
\end{eqnarray}

Then, (\ref{opti3'}) becomes
\begin{align}
&\max (G_3+J_3+K_3+L_3)\nonumber\\&=\max(K_3+2(N_1+2N_2-M_1-M_2))\nonumber\\&=\max(K_3)\\
\mathrm{st.} ~& K_3 \leq 2(N_1-N_2)+L_3~(\ref{R2})\nonumber\\ &G_3+J_3=2(N_1+2N_2-M_1-M_2)-L_3 ~(\ref{R1})\nonumber\\
      &K_3\leq G_3\leq N_1 ~(\ref{K3N1})~\mathrm{and}~ L_3\leq J_3\leq N_2 ~(\ref{J3c})\nonumber\\
      &M_1\geq M_2\geq N_1\geq N_2 ~\mathrm{and}~ M_1+M_2\leq3N_2\nonumber
      \end{align}

Since $K_3\leq G_3$ and $L_3\leq J_3$, we have $K_3+L_3\leq G_3+J_3$, which is equivalent to
\begin{eqnarray}
  K_3\leq G_3+J_3-L_3=2(N_1+2N_2-M_1-M_2)-2L_3 \label{cc}
\end{eqnarray}

Combining (\ref{R2}) and (\ref{cc}), we can get
\begin{eqnarray}
   K_3=\mathrm{min}\{2(N_1-N_2)+L_3,~ 2(N_1+2N_2-M_1-M_2)-2L_3\}\nonumber\\
   (75)~~~~\label{K_3}\nonumber
\end{eqnarray}
\setcounter{equation}{75}
Let $x=3N_2-M_1-M_2$, (\ref{K_3}) can be expressed as
\begin{eqnarray}
 K_3 =\left\{\begin{array}{cc}
                             2(N_1-N_2)+L_3 & \mathrm{if}~ L_3\leq\frac{2}{3}x \\
                             2(N_1+2N_2-M_1-M_2)-2L_3 & \mathrm{if}~ L_3\geq\frac{2}{3}x
                           \end{array}\right.\label{G3m}
 \end{eqnarray}

The problem becomes finding $L_3$ so that $K_3$ is maximized. Let $K_3^1 =2(N_1-N_2)+\lfloor\frac{2}{3}x\rfloor$ and $K_3^2 =2(N_1+2N_2-M_1-M_2)-2\lceil\frac{2}{3}x\rceil$, $K_3=\max\{K_3^1,~K_3^2\}$.
Since $K_3^2-K_3^1=2x-\lfloor\frac{2}{3}x\rfloor-2\lceil\frac{2}{3}x\rceil\leq0$, we can get $K_3=K_3^1=2(N_1-N_2)+\lfloor\frac{2}{3}x\rfloor$ and $L_3=\lfloor\frac{2}{3}x\rfloor$.

Then, since $G_3+J_3=2(N_1+2N_2-M_1-M_2)-\lfloor\frac{2}{3}x\rfloor$ (according to (\ref{cccc})) and $G_3\geq K_3,~J_3\geq L_3$, we can choose
\begin{eqnarray}
  G_3&=&2N_1+N_2-M_1-M_2-\lfloor\frac{x}{3}\rfloor\nonumber\\
  J_3&=&3N_2-M_1-M_2-\lceil\frac{x-1}{3}\rceil\nonumber
\end{eqnarray}

Note that $G_3-K_3=x-\lfloor\frac{x}{3}\rfloor-\lfloor\frac{2x}{3}\rfloor$ and $J_3-L_3=x-\lceil\frac{x-1}{3}\rceil-\lfloor\frac{2x}{3}\rfloor$, which can be expressed as
\begin{eqnarray}
\left\{\begin{array}{ll}
     G_3=K_3,~J_3=L_3 & \mathrm{when}~ x\mod 3=0\nonumber\\
    G_3=K_3+1,~J_3=L_3+1 & \mathrm{when}~ x\mod 3=1\nonumber\\
     G_3=K_3+1,~J_3=L_3 &\mathrm{when}~ x\mod 3=2\nonumber
   \end{array}\right.
\end{eqnarray}

In addition, since $M_1+M_2\geq N_1+N_2\geq 2N_2$, we have $2N_1+N_2-M_1-M_2\leq N_1$ and $3N_2-M_1-M_2\leq N_2$. Therefore, $G_3=2N_1+N_2-M_1-M_2-\lfloor\frac{x}{3}\rfloor\leq N_1$ and $J_3=3N_2-M_1-M_2-\lceil\frac{x-1}{3}\rceil\leq N_2$.

Finally, the length of each message block can be calculated as
\begin{eqnarray}
  Q_{11}&=&L_1+L_2+L_3=2(M_1-N_2)+\lfloor\frac{2}{3}x\rfloor\nonumber\\
  Q_{21}&=&K_1+K_2+K_3=2(M_1-N_2)+\lfloor\frac{2}{3}x\rfloor\nonumber\\
  Q_{12}&=&J_1+J_2+J_3=M_2+N_2-M_1-\lceil\frac{x-1}{3}\rceil\nonumber\\
  Q_{22}&=&G_1+G_2+G_3=M_2+N_2-M_1-\lfloor\frac{x}{3}\rfloor\nonumber
\end{eqnarray}

Now, we show that $Q_{rt}>0~~(r,t=1,2)$. Before the discussion, note that if $N_1=1$, then $M_2=N_1=N_2=1$ (due to $ M_1+M_2\leq3N_2$ and $M_1\geq M_2\geq N_1\geq N_2$). As a consequence, similar to subcase (1) of case $B$, $Q_{rt}>0~(r,t=1,2)$ can not be achieved with $N_1=1$. Hence, in case $C$ we focus on $N_1\geq2$. Also note that the case $M_1+M_2=3N_2=2N_1+N_2$ can be excluded from this case as it is already addressed in case $B$.

For $Q_{11}$ and $Q_{21}$, if $M_1>N_2$, then $Q_{11}=Q_{21}>0$ for sure. If $M_1=N_2$, then $M_1=M_2=N_1=N_2$ (since $M_1\geq M_2\geq N_1\geq N_2$) and $x=3N_2-M_1-M_2=N_2=N_1$. Hence, $Q_{11}=Q_{21}=\lfloor\frac{2}{3}x\rfloor\geq1$ ($N_1\geq 2$).

For $Q_{12}$, we have $Q_{12}\geq N_2+M_2-M_1-\frac{x+1}{3}=\frac{(2(2M_2-M_1)-1)}{3}$. Since $3M_2\geq3N_2\geq M_1+M_2$, $2M_2\geq M_1$. Since $M_1+M_2=3M_2=2N_1+N_2$ is excluded from case $C$, we can get $2M_2\neq M_1$ (if $2M_2=M_1$, then $M_2=N_2=N_1$). Therefore, in this case $2M_2-M_1\geq1$. Hence, $Q_{12}>0$ for sure as $2(2M_2-M_1)\geq2$.

For $Q_{22}$, we have $Q_{22}\geq M_2+N_2-M_1-\frac{x}{3}=\frac{4M_2-2M_1}{3}>0$.

The total achievable DoF can be calculated as
\begin{eqnarray}
  \frac{Q_{11}+Q_{21}+Q_{12}+Q_{22}}{2}=M_1+M_2-N_2+\frac{1}{2}\lfloor\frac{2}{3}x\rfloor\nonumber\\
  =\left\{\begin{array}{lc}
                                                \frac{2}{3}(M_1+M_2) & \mathrm{when}~x\mod 3=0 \\
                                                \frac{2}{3}(M_1+M_2)-\frac{1}{3} & \mathrm{when}~x\mod 3=1 \\
                                                 \frac{2}{3}(M_1+M_2)-\frac{1}{6} & \mathrm{when}~x\mod 3=2
                                              \end{array}\right.
\end{eqnarray}

Since the outer-bound DoF in this case is $\frac{2}{3}(M_1+M_2)$ \cite{MIMO}, we can see that the region of the gap between our achievable DoF and the outer-bound DoF is $0, \frac{1}{3}, \frac{1}{6} $ for $x\mod 3=0,~1,~2$, respectively.
\end{proof}

\begin{example} Three examples are given in this case, which are $(5,4,4,3)$ with $x\mod 3=0$, $(7,4,4,4)$ with $x\mod 3=1$, and $(7,6,6,5)$ with $x\mod 3=2$. For $(5,4,4,3)$, we get $L_1=L_2=2$ and $L_3=0$; $K_1=K_2=1$ and $K_3=2$; $J_1=J_2=1$ and $J_3=0$; $G_1=G_2=0$ and $G_3=2$. Totally 12 signals are transmitted, achieving DoF of $6$, which is equal to the outer bound. For $(7,4,4,4)$, we get $L_1=L_2=3$ and $L_3=0$; $K_1=K_2=3$ and $K_3=0$; $J_1=J_2=0$ and $J_3=1$; $G_1=G_2=0$ and $G_3=1$. Totally $14$ signals are transmitted, achieving DoF of $7$, while the outer bound is $\frac{2}{3}(M_1+M_2)=\frac{22}{3}$. For $(7,6,6,5)$, we get $L_1=L_2=2$ and $L_3=1$; $K_1=K_2=1$ and $K_3=3$; $J_1=J_2=1$ and $J_3=1$; $G_1=G_2=0$ and $G_3=4$. Totally $17$ signals are transmitted, achieving DoF of $\frac{17}{2}$, while the outer bound is $\frac{2}{3}(M_1+M_2)=\frac{26}{3}$.\end{example}

\begin{remark} Now, we justify the results of case $C$ intuitively as shown in Fig. \ref{c}.

\begin{figure}[t!]
  \begin{center}
        \includegraphics[width=1.1\columnwidth]{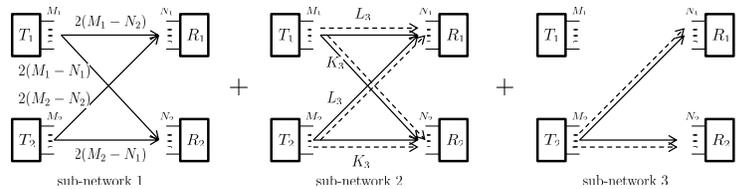}
        \caption{The three concatenated sub-networks (Case $C$)}
        \label{c}
    \end{center}
\end{figure}

At first, each link uses all the interference-free dimensions as shown in sub-network 1 of Fig. \ref{c}, no interference is caused on either receiver. After that, there are $2(M_1-N_2)+2(M_2-N_2)$ desired signals on $R_1$, and $2N_1-2(M_1-N_2)-2(M_2-N_2)=2(N_1+2N_2-M_1-M_2)$ real dimensions remaining unoccupied. On $R_2$, there are $2(M_1-N_1)+2(M_2-N_1)$ desired signals, and $2N_2-2(M_1-N_1)-2(M_2-N_1)=2(2N_1+N_2-M_1-M_2)$ unoccupied real dimensions. Note that this part is equivalent to equation (\ref{KLJ}).

Then, each link transmits some more messages with dimensions that cause interference to undesired receivers. Interference alignment should be applied to minimize the effect of interference on both receivers. Specifically, as shown in sub-network 2 of Fig. \ref{c}, each signal in link $T_1$--$R_1$ is aligned with one signal in link $T_2$--$R_1$ at receiver $R_2$ (This can be denoted by (\ref{L3c})). Each signal in link $T_2$--$R_2$ is aligned with one signal in link $T_1$--$R_2$ at receiver $R_1$ (This can be denoted by (\ref{G3})). Note that there may be more dimensions in link $T_2$--$R_1$ and link $T_1$--$R_2$ to be used ($G_3\geq K_3$ and $J_3\geq L_3$), but at this step we only pick those that are aligned with the signals in links $T_1$--$R_1$ and $T_2$--$R_2$.

As we can see, in sub-network $2$ the number of desired and interference signals on $R_1$ are $2L_3$ and $K_3$, respectively. On $R_2$, the number of desired and interference signals are $2K_3$ and $L_3$, respectively. Recall the dimensions left from sub-network $1$, we can get that
\begin{eqnarray}
  2L_3+K_3 \leq 2(N_1+2N_2-M_1-M_2) \label{2L3}\\
  2K_3+L_3 \leq 2(2N_1+N_2-M_1-M_2) \label{2G3}
\end{eqnarray}

$L_3$ and $K_3$ are determined by (\ref{2L3}) and (\ref{2G3}).

When $M_1+M_2\mod 3=0$ ($x\mod3=0$), it can be calculated that $L_3=2N_2-\frac{2}{3}(M_1+M_2)=\frac{2}{3}x$ and $K_3=2N_1-\frac{2}{3}(M_1+M_2)=2(N_1-N_2)+\frac{2x}{3}$. Note that the equalities hold for both (\ref{2L3}) and (\ref{2G3}), which means that all dimensions have been occupied on both receivers. Therefore, sub-network $3$ does not exist. Hence, $J_3=L_3=\frac{2}{3}x$ and $G_3=K_3=2(N_1-N_2)+\frac{2x}{3}$. The amount of signals in each link can then be calculated by combining sub-networks $1$ and $2$. Specifically, $Q_{11}=2(M_1-N_2)+L_3=2(M_1-N_2)+\frac{2}{3}x$, $Q_{21}=2(M_1-N_1)+K_3=\frac{4}{3}M_1-\frac{2}{3}M_2$, $Q_{12}=2(M_2-N_2)+J_3=2(M_2-N_2)+\frac{2}{3}x$ and $Q_{22}=2(M_2-N_1)+G_3=\frac{4}{3}M_2-\frac{2}{3}M_1$. The achievable DoF equals $\frac{2}{3}(M_1+M_2)$.

When $M_1+M_2\mod 3=1$ ($x\mod3=2$), we can get that $L_3=2N_2-\frac{2(M_1+M_2)+1}{3}=\lfloor\frac{2}{3}x\rfloor$ and $K_3=2N_1-\frac{2(M_1+M_2)+1}{3}$. Note that the equality does not hold for (\ref{2L3}) and (\ref{2G3}), which means there are still $2(N_1+2N_2-M_1-M_2)-2L_3-K_3=1$ and $2(2N_1+N_2-M_1-M_2)-2K_3-L_3=1$ dimensions left on $R_1$ and $R_2$, respectively. In this case, either link $T_2$--$R_2$ or link $T_2$--$R_1$ (not both) can transmit one more message. If we let link $T_2$--$R_2$ transmit, then $G_3=K_3+1=2N_1+1-\frac{2(M_1+M_2)+1}{3}$. Consequently, $R_2$ receives one more desired signal and $R_1$ receives one more interference signal. All dimensions are occupied.
Hence, the length of each message block can be calculated as $Q_{11}=2(M_1-N_2)+\lfloor\frac{2}{3}x\rfloor=\frac{4M_1-2M_2-1}{3}$, $Q_{21}=2(M_1-N_1)+K_3=\frac{4M_1-2M_2-1}{3}$, $Q_{12}=2(M_2-N_2)+J_3=\frac{4M_2-2M_1-1}{3}$ and $Q_{22}=2(M_2-N_1)+G_3=\frac{4M_2-2M_1+2}{3}$. The achievable DoF equals $\frac{2}{3}(M_1+M_2)-\frac{1}{6}$.

When $M_1+M_2\mod 3=2$ ($x\mod3=1$), we can let $L_3=2N_2-\frac{2(M_1+M_2)+2}{3}=\lfloor\frac{2}{3}x\rfloor$ and $K_3=2N_1-\frac{2M_1+2M_2+2}{3}$. Note that the equality does not hold for (\ref{2L3}) and (\ref{2G3}), which means there are still $2(N_1+2N_2-M_1-M_2)-2L_3-K_3=2$ and $2(2N_1+N_2-M_1-M_2)-2K_3-L_3=2$ dimensions left on $R_1$ and $R_2$, respectively. It implies that each receiver still has two dimensions unoccupied.  In this case, link $T_2$--$R_1$ and link $T_2$--$R_2$ each transmits one more signal to $R_1$ and $R_2$, respectively, as shown in sub-network $3$ of Fig. \ref{c}. Consequently, each receiver receives one more desired signal and one more interference signal that occupy two dimensions. Also, $G_3=K_3+1=2N_1-\frac{2M_1+2M_2-1}{3}$ and $J_3=L_3+1=2N_2-\frac{2(M_1+M_2)-1}{3}$. The length of each message block can be calculated as $Q_{11}=2(M_1-N_2)+L_3=\frac{4M_1-2M_2-2}{3}$, $Q_{21}=2(M_1-N_1)+K_3=\frac{4M_1-2M_2-2}{3}$, $Q_{12}=2(M_2-N_2)+J_3=\frac{4M_2-2M_1+1}{3}$ and $Q_{22}=2(M_2-N_1)+G_3=\frac{4M_2-2M_1+1}{3}$. Hence, the achievable DoF equals $\frac{2}{3}(M_1+M_2)-\frac{1}{3}$.\end{remark}


\section {$N_1\geq N_2 \geq M_1 \geq M_2$}\label{M<N}

In this section, we discuss the cases when the number of receiver antennas are larger than the number of transmitter antennas.  We employ an precoding scheme based on interference alignment to show that exactly symmetrical result can be achieved. 


\subsection{Design of Transmitted Signals}

To avoid confusion, we let $\mathbf{m}_{rt}'$ denote the message vectors and let $Q'_{rt}$ denote their corresponding length ($\mathbf{m}_{rt}'\in R^{Q'_{rt}\times1}$). Each message vector is divided into two groups, i.e.,
\begin{eqnarray}
   \mathbf{m}_{11}'=\left[\begin{array}{cc}
                    \underbrace{({\mathbf{m}_{11}^1}')^\mathbf{T}}_{L_1'} & \underbrace{({\mathbf{m}_{11}^2}')^\mathbf{T}}_{L_2'}
                  \end{array}\right]^\mathbf{T}\nonumber\\
 \mathbf{m}_{21}'=\left[\begin{array}{cc}
                    \underbrace{({\mathbf{m}_{21}^1}')^\mathbf{T}}_{K_1'} & \underbrace{({\mathbf{m}_{21}^2}')^\mathbf{T}}_{K_2'}
                  \end{array}\right]^\mathbf{T}\nonumber \\
  \mathbf{m}_{12}'=\left[\begin{array}{cc}
                    \underbrace{({\mathbf{m}_{12}^1}')^\mathbf{T}}_{J_1'} & \underbrace{({\mathbf{m}_{12}^2}')^\mathbf{T}}_{J_2'}
                  \end{array}\right]^\mathbf{T}\nonumber\\
  \mathbf{m}_{22}'=\left[\begin{array}{cc}
                    \underbrace{({\mathbf{m}_{22}^1}')^\mathbf{T}}_{G_1'} & \underbrace{({\mathbf{m}_{22}^2}')^\mathbf{T}}_{G_2'}
                  \end{array}\right]^\mathbf{T}
\end{eqnarray}

Accordingly, we have
\begin{eqnarray}
L_1'+L_2'=Q_{11}',~ K_1'+K_2'=Q_{21}'\nonumber\\ J_1'+J_2'=Q_{12}',~ G_1'+G_2'=Q_{22}'\label{sum}
\end{eqnarray}

If the signals on each receiver are independent of each other, the total achievable DoF of the system can be calculated as
\begin{eqnarray}
D_{sum}=\frac{L_1'+L_2'+K_1'+K_2'+J_1'+J_2'+G_1'+G_2'}{2}\label{Dsum2}
\end{eqnarray}

We let ${\mathbf{m}_{11}^1}'$, ${\mathbf{m}_{21}^1}'$, ${\mathbf{m}_{12}^1}'$ and ${\mathbf{m}_{22}^1}'$ be precoded with $[\mathbf{v}_{11}^1\cdots \mathbf{v}_{11}^{L_1'}]$, $[\mathbf{v}_{21}^1\cdots \mathbf{v}_{21}^{K_1'}]$, $[\mathbf{v}_{12}^1\cdots \mathbf{v}_{12}^{J_1'}]$ and $[\mathbf{v}_{22}^1\cdots \mathbf{v}_{22}^{G_1'}]$, respectively; while ${\mathbf{m}_{11}^2}'$, ${\mathbf{m}_{21}^2}'$, ${\mathbf{m}_{12}^2}'$ and ${\mathbf{m}_{22}^2}'$ are precoded with $[\mathbf{w}_{11}^1\cdots \mathbf{w}_{11}^{L_2'}]$, $[\mathbf{w}_{21}^1\cdots \mathbf{w}_{21}^{K_2'}]$, $[\mathbf{w}_{12}^1\cdots \mathbf{w}_{12}^{J_2'}]$ and $[\mathbf{w}_{22}^1\cdots \mathbf{w}_{22}^{G_2'}]$, respectively.

Therefore, the transmitted signals can be expressed as
\begin{eqnarray}
  \mathbf{x}_{11}' &=&\underbrace{\left[\begin{array}{ccc}
                     \mathbf{v}_{11}^1& \cdots & \mathbf{v}_{11}^{L_1'}
                  \end{array}\right]{\mathbf{m}_{11}^1}'}_{{\mathbf{x}_{11}^1}'}
                  +\underbrace{\left[\begin{array}{ccc}
                     \mathbf{w}_{11}^1& \cdots & \mathbf{w}_{11}^{L_2'}
                  \end{array}\right]{\mathbf{m}_{11}^2}'}_{{\mathbf{x}_{11}^2}'}\nonumber\\
   \mathbf{x}_{21}' &=&\underbrace{\left[\begin{array}{ccc}
                     \mathbf{v}_{21}^1& \cdots & \mathbf{v}_{21}^{K_1'}
                  \end{array}\right]{\mathbf{m}_{21}^1}'}_{{\mathbf{x}_{21}^1}'}
                  +\underbrace{\left[\begin{array}{ccc}
                     \mathbf{w}_{21}^1& \cdots & \mathbf{w}_{21}^{K_2'}
                  \end{array}\right]{\mathbf{m}_{21}^2}'}_{{\mathbf{x}_{21}^2}'}\nonumber\\
   \mathbf{x}_{12}' &=&\underbrace{\left[\begin{array}{ccc}
                     \mathbf{v}_{12}^1& \cdots & \mathbf{v}_{12}^{J_1'}
                  \end{array}\right]{\mathbf{m}_{12}^1}'}_{{\mathbf{x}_{12}^1}'}
                  +\underbrace{\left[\begin{array}{ccc}
                     \mathbf{w}_{12}^1& \cdots & \mathbf{w}_{12}^{J_2'}
                  \end{array}\right]{\mathbf{m}_{12}^2}'}_{{\mathbf{x}_{12}^2}'}\nonumber\\
   \mathbf{x}_{22}' &=&\underbrace{\left[\begin{array}{ccc}
                     \mathbf{v}_{22}^1& \cdots & \mathbf{v}_{22}^{G_1'}
                  \end{array}\right]{\mathbf{m}_{22}^1}'}_{{\mathbf{x}_{22}^1}'}
                  +\underbrace{\left[\begin{array}{ccc}
                     \mathbf{w}_{22}^1& \cdots & \mathbf{w}_{22}^{G_2'}
                  \end{array}\right]{\mathbf{m}_{22}^2}'}_{{\mathbf{x}_{22}^2}'}\nonumber
\end{eqnarray}

\subsection{Precoder Design and Constraints of Signal Independence}

Next, we present the design of the precoding vectors in this scenario based on the received signals.

On $R_1$, the received signals can be expressed as
\begin{eqnarray}
   \mathbf{y}_1'&=& \mathbf{H}_{11}({\mathbf{x}_{11}^1}'+{\mathbf{x}_{11}^2}')+\mathbf{H}_{12}({\mathbf{x}_{12}^1}'+{\mathbf{x}_{12}^2}')\nonumber\\
               &+&\underbrace{\mathbf{H}_{11}({\mathbf{x}_{21}^1}'+{\mathbf{x}_{21}^2}')+\mathbf{H}_{12}({\mathbf{x}_{22}^1}'+{\mathbf{x}_{22}^2}')}_{\rm{interferece}}+\mathbf{z}_1
 \end{eqnarray}

With asymmetric signaling, its real signal expression can be written as (omitting the noise)
\begin{eqnarray}
  {\bar{Y}_1}'=\nonumber \bar{H}_{11}[\bar{V}_{11}^1, \cdots,  \bar{V}_{11}^{L_1'}]{\mathbf{m}_{11}^1}'+\bar{H}_{11}[\bar{W}_{11}^1, \cdots, \bar{W}_{11}^{L_2'}]{\mathbf{m}_{11}^2}'\\
                  +\nonumber \bar{H}_{12}[\bar{V}_{12}^1, \cdots, \bar{V}_{12}^{J_1'}]{\mathbf{m}_{12}^1}'
                  +\bar{H}_{12}[\bar{W}_{12}^1, \cdots, \bar{W}_{12}^{J_2'}]{\mathbf{m}_{12}^2}'\\
                 +\underbrace{\bar{H}_{11}[\bar{V}_{21}^1, \cdots, \bar{V}_{21}^{K_1'}]{\mathbf{m}_{21}^1}'+\bar{H}_{11}[\bar{W}_{21}^1, \cdots, \bar{W}_{21}^{K_2'}]{\mathbf{m}_{21}^2}'}_{\mathrm{interference}}\nonumber\\
+\underbrace{\bar{H}_{12}[\bar{V}_{22}^1, \cdots, \bar{V}_{22}^{G_1'}]{\mathbf{m}_{22}^1}'
                  +\bar{H}_{12}[\bar{W}_{22}^1, \cdots, \bar{W}_{22}^{G_2'}]{\mathbf{m}_{22}^2}'}_{\mathrm{interference}}\nonumber
\end{eqnarray}
We want to align the signals in $\bar{H}_{11}[\bar{V}_{21}^1, \cdots, \bar{V}_{21}^{K_1'}]{\mathbf{m}_{21}^1}'$ and $\bar{H}_{12}[\bar{V}_{22}^1, \cdots, \bar{V}_{22}^{G_1'}]{\mathbf{m}_{22}^1}'$ one-to-one on $R_1$, which implies that $K_1'=G_1'$.
Specifically, we let $\bar{H}_{11}\bar{V}_{21}^i{m_{21}^i}'$ and $\bar{H}_{12}\bar{V}_{22}^i{m_{22}^i}'$ denote the $i$th signal of each group, respectively, and let
\begin{eqnarray}
   \bar{H}_{11}\bar{V}_{21}^i=\bar{H}_{12}\bar{V}_{22}^i=\mathbf{h}(i),~ i=1,2,\cdots, K_1'\label{h1}
\end{eqnarray}
where $\mathbf{h}(i)$ is the direction that the pair is aligned to on $R_1$.
As we can see, $\mathbf{h}(i)$, $\bar{V}_{21}^i$ and $\bar{V}_{22}^i$ can be calculated jointly as follows.
\begin{eqnarray}
  \underbrace{\left[\begin{array}{ccc}
    \mathbf{I} & -\bar{H}_{11} & \mathbf{0}\\
    \mathbf{I} & \mathbf{0} & -\bar{H}_{12}
  \end{array}\right]}_{\bar{H}'}\left[\begin{array}{c}
                      \mathbf{h}(i) \\
                      \bar{V}_{21}^i \\
                      \bar{V}_{22}^i
                    \end{array}\right]
   =\mathbf{0}\label{down}
\end{eqnarray}
where $\mathbf{I}\in \mathbb{R}^{2N_1\times2N_1}$, $\bar{H}'\in \mathbb{R}^{4N_1\times(2N_1+2M_1+2M_2)}$ and $\mathbf{h}(i)\in \mathbb{R}^{2N_1\times1}$.
This implies that each pair of signals can only be aligned onto one of some certain directions ($\mathbf{h}(i)$). The amount of these directions is equal to the number of independent column vectors of the null space of $\bar{H}'$. To guarantee the independence of the signals within the same group, the aligned signal pairs must be on different directions, which means
\begin{eqnarray}
   K_1'=G_1'\leq \dim(\ker{\bar{H}'})=\max\{2M_1+2M_2-2N_1,~0\}\label{KK1}
\end{eqnarray}
where $\dim(\ker{\bar{H}'})$ denotes the number of dimensions of the kernel of $\bar{H}'$, i.e., the nullity of $\bar{H}'$.
Hence, the precoders $[\bar{V}_{21}^1, \cdots, \bar{V}_{21}^{K_1'}]$ and $[\bar{V}_{22}^1, \cdots, \bar{V}_{22}^{G_1'}]$ can be designed together.

On $R_2$, the received signals are
\begin{eqnarray}
   \mathbf{y}_2'&=& \mathbf{H}_{21}({\mathbf{x}_{21}^1}'+{\mathbf{x}_{21}^2}')+\mathbf{H}_{22}({\mathbf{x}_{22}^1}'+{\mathbf{x}_{22}^2}')\nonumber\\
               &+&\underbrace{\mathbf{H}_{21}({\mathbf{x}_{11}^1}'+{\mathbf{x}_{11}^2}')+\mathbf{H}_{22}({\mathbf{x}_{12}^1}'+{\mathbf{x}_{12}^2}')}_{\rm{interferece}}+\mathbf{z}_1
 \end{eqnarray}
 With asymmetric signaling, its real signal expression can be written as (omitting the noise)
\begin{eqnarray}
  {\bar{Y}_2}'=\nonumber \bar{H}_{21}[\bar{V}_{21}^1, \cdots, \bar{V}_{21}^{K_1'}]{\mathbf{m}_{21}^1}'+\bar{H}_{21}[\bar{W}_{21}^1, \cdots, \bar{W}_{21}^{K_2'}]{\mathbf{m}_{21}^2}'\\
                  +\nonumber \bar{H}_{22}[\bar{V}_{22}^1, \cdots, \bar{V}_{22}^{G_1'}]{\mathbf{m}_{22}^1}'
                  +\bar{H}_{22}[\bar{W}_{22}^1, \cdots, \bar{W}_{22}^{G_2'}]{\mathbf{m}_{22}^2}'\\
                 +\underbrace{\bar{H}_{21}[\bar{V}_{11}^1, \cdots,  \bar{V}_{11}^{L_1'}]{\mathbf{m}_{11}^1}'+\bar{H}_{21}[\bar{W}_{11}^1, \cdots, \bar{W}_{11}^{L_2'}]{\mathbf{m}_{11}^2}'}_{\mathrm{interference}}\nonumber\\
+\underbrace{\bar{H}_{22}[\bar{V}_{12}^1, \cdots, \bar{V}_{12}^{J_1'}]{\mathbf{m}_{12}^1}'
                  +\bar{H}_{22}[\bar{W}_{12}^1, \cdots, \bar{W}_{12}^{J_2'}]{\mathbf{m}_{12}^2}'}_{\mathrm{interference}}\nonumber
\end{eqnarray}
We want to align the signals in $\bar{H}_{21}[\bar{V}_{11}^1, \cdots, \bar{V}_{11}^{L_1'}]{\mathbf{m}_{11}^1}'$ and $\bar{H}_{22}[\bar{V}_{12}^1, \cdots, \bar{V}_{12}^{J_1'}]{\mathbf{m}_{12}^1}'$ one-to-one on $R_2$. Likewise, we can get
\begin{eqnarray}
  \underbrace{\left[\begin{array}{ccc}
    \mathbf{I} & -\bar{H}_{21} & \mathbf{0}\\
    \mathbf{I} & \mathbf{0} & -\bar{H}_{22}
  \end{array}\right]}_{\bar{H}''}\left[\begin{array}{c}
                      \mathbf{h}'(i) \\
                      \bar{V}_{11}^i \\
                      \bar{V}_{12}^i
                    \end{array}\right]
   =\mathbf{0}\label{down2}
\end{eqnarray}
where $\mathbf{I}\in \mathbb{R}^{2N_2\times2N_2}$, $\bar{H}''\in \mathbb{R}^{4N_2\times(2N_2+2M_1+2M_2)}$ and $\mathbf{h}'(i)\in \mathbb{R}^{2N_2\times1}$.
Accordingly, we have
\begin{eqnarray}
   \bar{H}_{21}\bar{V}_{11}^i=\bar{H}_{22}\bar{V}_{12}^i=\mathbf{h}'(i).~ i=1,2,\cdots, L_1'~~~~~~~~~~~~~~~~\label{h2}\\
   L_1'=J_1'\leq \dim(\ker{\bar{H}''})=\max\{2M_1+2M_2-2N_2,~0\}\label{LL1}
\end{eqnarray}
Hence, the precoders $[\bar{V}_{11}^1, \cdots, \bar{V}_{11}^{L_1'}]$ and $[\bar{V}_{12}^1, \cdots, \bar{V}_{12}^{J_1'}]$ are determined.

Next, we shall design other four groups of precoders. The design principle is to guarantee the signals on each receiver to be independent of each other.

We first examine the received signals at $R_1$. The real version of the signals from transmitter $T_1$ are $\bar{H}_{11}\bar{V}_{21}\mathbf{m}_{21}^{1'}$,
$\bar{H}_{11}\bar{V}_{11}\mathbf{m}_{11}^{1'}$, $\bar{H}_{11}\bar{W}_{21}\mathbf{m}_{21}^{2'}$ and $\bar{H}_{11}\bar{W}_{11}\mathbf{m}_{11}^{2'}$,
which can be expressed as $\bar{H}_{11}\left[\begin{array}{cccc}
                                                                         \bar{V}_{21} & \bar{V}_{11}& \bar{W}_{21} & \bar{W}_{11}
                                                                       \end{array}\right]$.
                                                                        Note that $[\bar{V}_{21}^1, \cdots, \bar{V}_{21}^{K_1'}]$
and $[\bar{V}_{11}^1, \cdots, \bar{V}_{11}^{L_1'}]$ are designed according to (\ref{down}) and (\ref{down2}), respectively. Therefore,
$\left[\begin{array}{cc}
                                                                         \bar{V}_{21} & \bar{V}_{11}
                                                                       \end{array}\right]$
has full column rank almost for sure due to the channel randomness. Then, we design $[\bar{W}_{21}^1, \cdots, \bar{W}_{21}^{K_2'}]$ and $[\bar{W}_{11}^1, \cdots, \bar{W}_{11}^{L_2'}]$ so that $\left[\begin{array}{cccc}
\bar{V}_{21} & \bar{V}_{11}& \bar{W}_{21} & \bar{W}_{11}
 \end{array}\right]\in\mathbb{R}^{2M_1\times(K_1'+K_2'+L_1'+L_2')}$ has full column rank. As we can see, the precoders exist as long as the number of signals are no more than the number of real dimensions of $T_1$, i.e.,
\begin{eqnarray}
  L_1'+L_2'+K_1'+K_2'= Q'_{11}+Q'_{21}\leq2M_1 \label{transmitter1}
  \end{eqnarray}
Since $M_1\leq N_1$, the received signals from $T_1$, $\bar{H}_{11}\left[\begin{array}{cccc}
                                                                         \bar{V}_{21} & \bar{V}_{11}& \bar{W}_{21} & \bar{W}_{11}
                                                                       \end{array}\right]\in\mathbb{R}^{2N_1\times(K_1'+K_2'+L_1'+L_2')}$, also has full column rank for sure.
Further, the real version of the signals from transmitter $T_2$ can be expressed as $\bar{H}_{12}\left[\begin{array}{cccc}
                                                                         \bar{V}_{12} & \bar{V}_{22}& \bar{W}_{12} & \bar{W}_{22}
                                                                       \end{array}\right]$. Similarly, $\bar{W}_{12}$ and $\bar{W}_{22}$ can be found to guarantee the full column rank as long as
\begin{eqnarray}
 J_1'+J_2'+G_1'+G_2'=Q'_{12}+Q'_{22}\leq2M_2 \label{transmitter2}
\end{eqnarray}
Finally, the total received signals on $R_1$ can be expressed in real version as \\$\left[
 \begin{array}{cc}
 \bar{H}_{11}(\bar{V}_{11}~\bar{W}_{11}~\bar{V}_{21}~\bar{W}_{21}) &  \bar{H}_{12}(\bar{V}_{12}~\bar{W}_{12}~\bar{W}_{22}) \\
 \end{array}
 \right]$ ($\bar{H}_{12}\bar{V}_{22}$ is aligned with $\bar{H}_{11}\bar{V}_{21}$). Based on above discussion and the property of random channels, the full column rank of the matrix can be guaranteed as long as
 \begin{eqnarray}
L_1'+L_2'+J_1'+J_2'+K_1'+K_2'+G_2'\leq 2N_1\label{N1}
\end{eqnarray}

Next, we examine the received signals at $R_2$. The real version of the signals from transmitters $T_1$ and $T_2$ can be expressed as $\bar{H}_{21}\left[\begin{array}{cccc}
                                                                         \bar{V}_{21} & \bar{V}_{11}& \bar{W}_{21} & \bar{W}_{11}
                                                                       \end{array}\right]$ and $\bar{H}_{22}\left[\begin{array}{cccc}
                                                                         \bar{V}_{12} & \bar{V}_{22}& \bar{W}_{12} & \bar{W}_{22}
                                                                       \end{array}\right]$, respectively. They both have full column rank if (\ref{transmitter1}) and  (\ref{transmitter2}) are satisfied.
Then, the total received signals can be expressed in real version as $\left[
 \begin{array}{cc}
 \bar{H}_{21}(\bar{V}_{21}~\bar{W}_{21}~\bar{V}_{11}~\bar{W}_{11}) &  \bar{H}_{22}(\bar{V}_{22}~\bar{W}_{22}~\bar{W}_{12}) \\
 \end{array}
 \right]$ ($\bar{H}_{22}\bar{V}_{12}$ is aligned with $\bar{H}_{21}\bar{V}_{11}$). The full column rank can be guaranteed if
 \begin{eqnarray}
K_1'+K_2'+G_1'+G_2'+L_1'+L_2'+J_2'\leq 2N_2\label{N2}
\end{eqnarray}

Therefore, the constraints of signal independence are (\ref{KK1}) and (\ref{LL1})-(\ref{N2}).

\subsection{Achievable DoF}

Next, we investigate the achievable DoF when $M_t\leq N_r$. According to the antenna configurations, one can note that for each case in Tables \ref{table2}, \ref{table3} and \ref{table4}, there is a symmetrical one in this scenario. By swapping $M_1$ and $N_1$, $M_2$ and $N_2$, and letting $Q'_{rt}=Q_{tr}$ (in Table \ref{table2}, \ref{table3}, \ref{table4}), we can get Tables \ref{table6}, \ref{table7} and \ref{table8}. Next, we prove that the results in Tables \ref{table6}, \ref{table7} and \ref{table8} satisfy all the constraints of independence and are achievable with our scheme.

\begin{theorem} In $2\times2$ MIMO X network with $M_t$ antennas at transmitter $t$ and $N_r$ antennas at receiver $r$, when $N_1\geq N_2\geq M_1\geq M_2$ and $3M_2<N_1+N_2<2M_1+M_2$, the total achievable DoF of the network is $\frac{N_1+N_2+M_2}{2}$ (the outer-bound). The length of each message block is shown in Table \ref{table6}.\end{theorem}

\begin{table*}[t!]
\begin{center}\caption{Length of Message Vectors in Case $A'$ ($3M_2<N_1+N_2<2M_1+M_2$)}
\renewcommand{\arraystretch}{2 }
\centering\label{table6}
\begin{tabular} {c |c |c |c |c}
\hline
$Q'_{11}$  & $Q'_{12}$ & $Q'_{21}$ & $Q'_{22}$ & Achievable DoF\\
 \hline
 $N_1-N_2+M_2$ & $N_1-N_2+M_2$ &$2(N_2-M_2)$ & $N_2-N_1+M_2$ & $\frac{N_1+N_2+M_2}{2}$\\
\hline\hline
\end{tabular}
\end{center}
\end{table*}

\begin{proof} Note that this case is symmetrical to Case $A$ of Section \ref{A}. Therefore, we swap $N_r$ and $M_r$ and let $Q_{rt}=Q'_{tr}$ in Table \ref{table2}. As a result, the length of each message block in this case can be written as
\begin{eqnarray}
  Q'_{11} &=& L_1'+L_2'=N_1-N_2+M_2 \nonumber\\
  Q'_{12} &=& J_1'+J_2'=N_1-N_2+M_2 \nonumber\\
  Q'_{21} &=& K_1'+K_2'=2(N_2-M_2) \nonumber\\
  Q'_{22} &=& G_1'+G_2'=N_2+M_2-N_1\label{Q'}
\end{eqnarray}

Since $M_1>M_2$ ($3M_2<N_1+N_2<2M_1+M_2$), we have $N_2\geq M_1>M_2$. Hence, $ Q'_{21}>0$. Since $N_2+M_2\geq M_1+M_2> N_1$, we have $Q'_{22}>0$ for sure.

Next, we show that based on our proposed scheme, a proper value for each parameter can be found in (\ref{Q'}) while satisfying all the constraints of independence (\ref{KK1}) and (\ref{LL1})-(\ref{N2}).

First of all, it can be proved that (\ref{transmitter1}) and (\ref{transmitter2}) are satisfied by (\ref{Q'}).

Then, since $M_1+M_2>N_1+N_2-M_1\geq N_1\geq N_2$, based on (\ref{KK1}) and (\ref{LL1}) we have
\begin{eqnarray}
 G_1'=K_1'\leq\dim(\ker{\bar{H}'})=2M_1+2M_2-2N_1\nonumber\\
   L_1'=J_1'\leq\dim(\ker{\bar{H}''})= 2M_1+2M_2-2N_2\nonumber
\end{eqnarray}

Accordingly, we can choose
\begin{eqnarray}
G_1'&=&N_2-N_1+M_2,~\mathrm{and}~G_2'=0\nonumber\\
 K_1'&=& N_2-N_1+M_2,~\mathrm{and}~ K_2'=N_1+N_2-3M_2\nonumber\\
 L_1'&=& J_1'=N_1-N_2+M_2,~\mathrm{and}~L_2'=J_2'=0\nonumber
\end{eqnarray}

Since $N_1+N_2<2M_1+M_2$,  it can be proved that $K_1'=N_2-N_1+M_2\leq2M_1+2M_2-2N_1$ and $L_1'=N_1-N_2+M_2\leq2M_1+2M_2-2N_2$. Hence, (\ref{KK1}) and (\ref{LL1}) are satisfied.

It can be also proved that the constraints (\ref{N1}) and (\ref{N2}) are satisfied as well. Therefore, the DoF can be calculated with (\ref{Dsum2}) and is equal to $\frac{N_1+N_2+M_2}{2}$.
\end{proof}

\begin{theorem} In $2\times2$ MIMO X network with $M_t$ antennas at transmitter $t$ and $N_r$ antennas at receiver $r$, when $N_1\geq N_2\geq M_1\geq M_2$ and $N_1+N_2\geq2M_1+M_2$, the achievable DoF equals\\

~~~~~~~~~~~$\left\{\begin{array}{cc}
           M_1+M_2-1 & \mathrm{if}~ N_1\geq N_2=M_1=M_2\\
 M_1+M_2-\frac{1}{2} & \mathrm{if}~ N_1\geq N_2=M_1>M_2\\
     M_1+M_2  & \mathrm{if} ~N_1\geq N_2>M_1\geq M_2
   \end{array}\right.$\\

The length of each message block in different cases is shown in Table \ref{table7}.\end{theorem}

\begin{table*}
[t!]
\begin{center}\caption{Length of Message Vectors in Case $B'$ ($N_1+N_2\geq2M_1+M_2$)}
\renewcommand{\arraystretch}{2 }
\centering\label{table7}
\begin{tabular} {c |c |c |c| c| c}
\hline \hline
 & $Q'_{11}$  & $Q'_{12}$ & $Q'_{21}$ & $Q'_{22}$ & Achievable DoF\\
\hline
(1)$N_1\geq N_2=M_1=M_2$ &$ 2M_2-2 $ & $ 2M_2-2 $& $1$ & $1$ & $M_1+M_2-1$\\
\hline
(2)$N_1\geq N_2=M_1>M_2 $ &  $2M_2-1$ &$ 2M_2-1$ &$2(N_2-M_2)$ & $1$ & $M_1+M_2-\frac{1}{2}$\\
\hline
(3)$N_1\geq N_2>M_1\geq M_2$ & $2M_1-Q'_{21}$   & $2M_2-Q'_{22}$ & $\min\{2(N_2-M_2)~,M_1\}$ &$\min\{2(N_2-M_1)~,M_2\}$ & $M_1+M_2$\\
\hline\hline
\end{tabular}
\end{center}
\hrulefill
\end{table*}

\begin{proof} When $N_2=M_1\geq M_2$ ((1) and (2) of Table \ref{table7}), the cases are symmetrical to (1) and (2) of case $B$, respecitvely. In addition, since $N_1+N_2\geq2M_1+M_2$, we have $N_1\geq M_1+M_2$. According to (\ref{KK1}) and (\ref{LL1}), we can get $ \dim(\ker{\bar{H}'})=0$ and $\dim(\ker{\bar{H}''})=2M_2$.

Also, besides the independence constraints, $Q_{rt}'\geq 1$ needs to be taken into consideration as well.

Therefore, for $N_1\geq N_2=M_1=M_2 $ ((1) of Table \ref{table7}), we can choose
\begin{eqnarray}
  \begin{array}{ll}
    K_1'= 0,~K_2'=1 &  \mathrm{and}~Q_{21}'=K_1'+K_2'=1\\
    G_1'= 0,~G_2'=1 & \mathrm{and}~ Q_{22}'=G_1'+G_2'=1 \\
    L_1'= 2M_2-2,~L_2'=0 & \mathrm{and}~Q_{11}'=L_1'+L_2'=2M_2-2 \\
    J_1'= 2M_2-2,~J_2'=0 & \mathrm{and}~Q_{12}'=J_1'+J_2'=2M_2-2
  \end{array}\nonumber
\end{eqnarray}

Note that $Q_{rt}'>0~(r,t=1,2)$ can only be achieved when $N_2=M_1=M_2>1$.

For $N_1\geq N_2=M_1>M_2$ ((2) of Table \ref{table7}), we can choose
\begin{eqnarray}
  \begin{array}{ll}
    K_1'= 0,~K_2'=2(N_2-M_2) &  \mathrm{and}~Q_{21}'=K_1'+K_2'=2(N_2-M_2) \\
    G_1'= 0,~G_2'=1 & \mathrm{and}~ Q_{22}'=G_1'+G_2'=1 \\
    L_1'= 2M_2-1,~L_2'=0 & \mathrm{and}~Q_{11}'=L_1'+L_2'=2M_2-1 \\
    J_1'= 2M_2-1,~J_2'=0 & \mathrm{and}~Q_{12}'=J_1'+J_2'=2M_2-1
  \end{array}\nonumber
\end{eqnarray}

It can be proved that all the constraints ((\ref{KK1}) and (\ref{LL1})-(\ref{N2})) are satisfied with above settings.

When $N_2>M_1$ ((3) of Table \ref{table7}), the outer-bound DoF can be achieved. Note that in this scenario, for a certain ($M_1,~M_2$), the outer-bound is fixed as $M_1+M_2$ (unrelated to $N_1$, $N_2$). Given a fixed transmitter antenna configuration ($M_1,~M_2$), the number of receiver antennas would satisfy either $N_1+N_2=2M_1+M_2$ or $N_1+N_2>2M_1+M_2$. If the one with $N_1+N_2=2M_1+M_2$ can achieve the outer-bound, it is obviously that those with $N_1+N_2>2M_1+M_2$ can also achieve the same outer-bound for sure.
Therefore, in this case we only need to show that the outer-bound can be achieved when $N_1+N_2=2M_1+M_2$.

Firstly, based on its symmetrical case ((3) of Table \ref{table3}), we can get the length of each message block as
\begin{eqnarray}
  Q'_{11} &=& L_1'+L_2'=2M_1-Q'_{21} \nonumber\\
  Q'_{12} &=& J_1'+J_2'=2M_2-Q'_{22} \nonumber\\
  Q'_{21} &=& K_1'+K_2'=\min\{2(N_2-M_2),~M_1\} \nonumber\\
  Q'_{22} &=& G_1'+G_2'=\min\{2(N_2-M_1)~,M_2\}=2(N_2-M_1)\nonumber
\end{eqnarray}

Note that since $N_1+N_2=2M_1+M_2$, $2(N_2-M_1)\leq N_1+N_2-2M_1=M_2$. It can be proved that (\ref{transmitter1}) and (\ref{transmitter2}) are satisfied.

Then, based on (\ref{KK1}) and (\ref{LL1}) we have $ \dim(\ker{\bar{H}'})=2(N_2-M_1)$ and $\dim(\ker{\bar{H}''})=2(N_1-M_1)$. Accordingly, we can choose
\begin{eqnarray}
  \left\{
    \begin{array}{cl}
       K_1'= 2(N_2-M_1), & \mathrm{and}~K_2'=2(M_1-M_2) \\
       G_1'= 2(N_2-M_1), & \mathrm{and}~G_2'=0 \\
       L_1'= 2(N_1-M_1), & \mathrm{and}~L_2'=0 \\
       J_1'= 2(N_1-M_1), & \mathrm{and}~J_2'=0 \\
    \end{array}
  \right.\nonumber
\end{eqnarray}
if $M_1>2(N_2-M_2)$, or
\begin{eqnarray}
  \left\{
    \begin{array}{ll}
       K_1'= 2(N_2-M_1), & \mathrm{and}~K_2'=3M_1-2N_2 \\
       G_1'= 2(N_2-M_1), & \mathrm{and}~G_2'=0 \\
       L_1'= M_2, & \mathrm{and}~L_2'=M_1-M_2\\
       J_1'= M_2, & \mathrm{and}~J_2'=N_1-N_2 \\
    \end{array}
  \right.\nonumber
\end{eqnarray}
if $M_1\leq2(N_2-M_2)$.

Note that $K_2'=3M_1-2N_2\geq0$ as $3M_1\geq 2M_1+M_2=N_1+N_2\geq2N_2$. It can be proved that (\ref{KK1}), (\ref{LL1}), (\ref{N1}) and (\ref{N2}) are all satisfied.

Finally, the achievable DoF can be calculated with (\ref{Dsum2}) and the result is equal to the outer-bound.
\end{proof}

\begin{theorem} In $2\times2$ MIMO X network with $M_t$ antennas at transmitter $t$ and $N_r$ antennas at receiver $r$, when $N_1\geq N_2\geq M_1\geq M_2$ and $N_1+N_2\leq3M_2$, the achievable DoF equals

~~~~~~~~~~~~~~~~~~$\left\{\begin{array}{cc}
      \frac{2(N_1+N_2)}{3} & \mathrm{if}~ x'\mod 3=0\\
     \frac{2(N_1+N_2)-1}{3} & \mathrm{if}~ x'\mod 3=1\\
     \frac{4(N_1+N_2)-1}{6}  &\mathrm{if}~ x'\mod 3=2
   \end{array}\right.$\\
where $x'=3M_2-N_1-N_2$.
The length of each message block is shown in Table \ref{table8}.\end{theorem}

\begin{proof} Based on its symmetrical case (case $C$ in Section \ref{C}), the length of each message block can be written as
\begin{eqnarray}
  Q_{11}'&=&L_1'+L_2'=2(N_1-M_2)+\lfloor\frac{2}{3}x'\rfloor\nonumber\\
  Q_{21}'&=&K_1'+K_2'=N_2+M_2-N_1-\lceil\frac{x'-1}{3}\rceil\nonumber\\
  Q_{12}'&=&J_1'+J_2'=2(N_1-M_2)+\lfloor\frac{2}{3}x'\rfloor\nonumber\\
  Q_{22}'&=&G_1'+G_2'=N_2+M_2-N_1-\lfloor\frac{x'}{3}\rfloor\nonumber
\end{eqnarray}

For the signals transmitted from $T_1$, since $Q_{11}'+Q_{21}'\leq N_1+N_2-M_2+\frac{x'}{3}=\frac{2}{3}(N_1+N_2)\leq 2M_2\leq 2M_1$, (\ref{transmitter1}) are satisfied.

For the signals transmitted from $T_2$, when $x'\mod 3=0~\mathrm{or}~1$, $Q_{12}'+Q_{22}'\leq N_1+N_2-M_2+\frac{x'}{3}=\frac{2}{3}(N_1+N_2)\leq 2M_2$. When $x'\mod 3=2$, $N_1+N_2<3M_2$, and $Q_{12}'+Q_{22}'= N_1+N_2-M_2+\frac{x'+1}{3}= \frac{2}{3}(N_1+N_2+1)\leq 2M_2$. Hence, (\ref{transmitter2}) are satisfied.

The proof of $Q_{rt}'>0~(r,t=1,2)$ is similar to that of Case $C$.

Then, according to (\ref{KK1}) and (\ref{LL1}), we have $\dim(\ker{\bar{H}'})=2M_1+2M_2-2N_1$ and $\dim(\ker{\bar{H}''})=2M_1+2M_2-2N_2$. The parameters can be chosen as follows.

When $x'\mod3=0$,
\begin{eqnarray}
\left\{
  \begin{array}{l}
    K_1'=2(N_2-M_2)+\frac{2x'}{3}=N_2-N_1+M_2-\frac{x'}{3} \\
    K_2'=0\\
    G_1'= N_2-N_1+M_2-\frac{x'}{3}, \mathrm{and}~G_2'=0 \\
    L_1'=2(N_1-M_2)+\frac{2x'}{3}=N_1-N_2+M_2-\frac{x'}{3} \\
    L_2'=0\\
    J_1'= N_1-N_2+M_2-\frac{x'}{3}, \mathrm{and}~J_2'=0 \\
  \end{array}
\right.\nonumber
\end{eqnarray}

When $x'\mod3=1$,
\begin{eqnarray}
\left\{
  \begin{array}{l}
    K_1'=2(N_2-M_2)+\lfloor\frac{2x'}{3}\rfloor+1=N_2-N_1+M_2-\lfloor\frac{x'}{3}\rfloor\\
    K_2'=0 \\
    G_1'= N_2-N_1+M_2-\lfloor\frac{x'}{3}\rfloor, \mathrm{and}~G_2'=0 \\
    L_1'=2(N_1-M_2)+\lfloor\frac{2x'}{3}\rfloor=N_1-N_2+M_2-\lfloor\frac{x'}{3}\rfloor-1\\
    L_2'=0 \\
    J_1'= N_1-N_2+M_2-\lfloor\frac{x'}{3}\rfloor-1,\mathrm{and}~J_2'=0 \\
  \end{array}
\right.\nonumber
\end{eqnarray}

When $x'\mod3=2$,
\begin{eqnarray}
 \left\{
  \begin{array}{l}
    K_1'=2(N_2-M_2)+\lfloor\frac{2x'}{3}\rfloor=N_2-N_1+M_2-\lceil\frac{x'}{3}\rceil\\
    K_2'=0 \\
    G_1'=N_2-N_1+M_2-\lceil\frac{x'}{3}\rceil, \mathrm{and}~G_2'=1 \\
    L_1'=2(N_1-M_2)+\lfloor\frac{2x'}{3}\rfloor=N_1-N_2+M_2-\lfloor\frac{x'}{3}\rfloor\\
    L_2'=0 \\
    J_1'= N_1-N_2+M_2-\lfloor\frac{x'}{3}\rfloor,\mathrm{and}~J_2'=0 \\
  \end{array}
\right.\nonumber
\end{eqnarray}


Note that since $N_1+N_2\leq 3M_2\leq 2M_1+M_2$, it can be proved that $K_1'=G_1'\leq N_2-N_1+M_2 \leq \dim(\ker{\bar{H}'})$ and $L_1'=J_1'\leq N_1-N_2+M_2\leq \dim(\ker{\bar{H}''})$. So (\ref{KK1}) and (\ref{LL1}) are both satisfied.

In addition, it can also be calculated that the constraints (\ref{N1}) and (\ref{N2}) are satisfied as well.

Finally, the achievable DoF equals $\frac{Q_{11}'+Q_{21}'+Q_{12}'+Q_{22}'}{2}=N_1+N_2-M_2+\frac{1}{2}\lfloor\frac{2x'}{3}\rfloor$. \end{proof}

Therefore, all cases have been proved to be symmetrical with the cases of $M_t\geq N_r$ scenario.

\begin{table*}
[t!]
\begin{center}\caption{Length of Message Vectors in Case $C'$ ($N_1+N_2\leq3M_2$)}
\renewcommand{\arraystretch}{2 }
\centering\label{table8}
\begin{tabular} {c |c |c |c| c}
\hline \hline
 $Q'_{11}$  & $Q'_{12}$ & $Q'_{21}$ & $Q'_{22}$ & Achievable DoF\\
 \hline
 $2(N_1-M_2)+\lfloor\frac{2x'}{3}\rfloor$   & $ 2(N_1-M_2)+\lfloor\frac{2}{3}x'\rfloor$ & $N_2+M_2-N_1-\lceil\frac{x'-1}{3}\rceil$ &$N_2+M_2-N_1-\lfloor\frac{x'}{3}\rfloor$ & $N_1+N_2-M_2+\frac{1}{2}\lfloor\frac{2x'}{3}\rfloor$\\
\hline\hline
\end{tabular}
\end{center}
\hrulefill
\end{table*}

\section{Conclusion}\label{con}

The achievable DoF of $2\times2$ MIMO X network is investigated. In the scenario of $M_t\geq N_r~(r,t=1,2)$, it is divided into three cases based on different types of antenna configurations. A practical asymmetric interference alignment and cancelation scheme was proposed that achieves outer-bound or near outer-bound DoF in each case. In addition, a thorough intuitive explanation was presented for each case to verify the result.
In the scenario of $M_t\leq N_r~(r,t=1,2)$, an interference alignment-based precoding scheme is utilized to show that the results are exactly symmetrical to the scenario of $M_t\geq N_r~(r,t=1,2)$.


\end{document}